\documentclass[letterpaper,twocolumn,american,showpacs,pr,aps,superscriptaddress,eqsecnum,nofootinbib]{revtex4}
\usepackage{amsfonts}%
\usepackage{amsmath}%
\usepackage{amssymb}%
\usepackage{graphicx}
\usepackage{color}
\usepackage{lipsum}
\usepackage{comment}
\usepackage{hyperref}
\providecommand{\U}[1]{\protect\rule{.1in}{.1in}}
\newtheorem{theorem}{Theorem}

\newtheorem{definition}[theorem]{Definition}

\newtheorem{proposition}[theorem]{Proposition}

\newenvironment{proof}[1][Proof]{\noindent\textbf{#1.} }{\ \rule{0.5em}{0.5em}}

\newcommand{\bes} {\begin{subequations}}
\newcommand{\ees} {\end{subequations}}
\newcommand{\bea} {\begin{eqnarray}}
\newcommand{\eea} {\end{eqnarray}}
\newcommand{\beq}{\begin{equation}}
\newcommand{\eeq}{\end{equation}}
\newcommand{\bal}{\begin{align}}
\newcommand{\eal}{\end{align}}

\def\>{\rangle}
\def\<{\langle}

\renewcommand{\min}{\textrm{min}}
\renewcommand{\max}{\textrm{max}}

\newcommand{\ignore}[1]{}

\def\TI{\textrm{TI}}

\definecolor{nblue}{rgb}{0.2,0.2,0.7}
\definecolor{ngreen}{rgb}{0.2,0.6,0.2}
\definecolor{nred}{rgb}{0.7,0.2,0.2}
\definecolor{nblack}{rgb}{0,0,0}

\begin{document}

\title{Quantum speed limits, coherence and asymmetry}

\author{Iman Marvian}
\affiliation{Research Laboratory of Electronics, Massachusetts Institute of Technology, Cambridge, MA 02139}
\affiliation{Department of Physics and Astronomy, Center for Quantum Information Science and Technology, University of Southern California, Los Angeles, CA 90089}
\author{Robert W. Spekkens}
\affiliation{Perimeter Institute for Theoretical Physics, 31 Caroline St. N, Waterloo, \\
Ontario, Canada N2L 2Y5}

\author{Paolo Zanardi}
\affiliation{Department of Physics and Astronomy, Center for Quantum Information Science and Technology, University of Southern California, Los Angeles, CA 90089}

\date{\today}

\begin{abstract}

The resource theory of asymmetry is a framework for classifying and quantifying the symmetry-breaking properties of both states and operations  relative to a given symmetry. In the special case where the symmetry is the set of translations generated by a fixed observable, asymmetry can be interpreted as coherence relative to the observable eigenbasis, and the resource theory of asymmetry provides a framework to study this notion of coherence. We here show that this notion of coherence naturally arises in the context of quantum speed limits. Indeed, the very concept of \emph{speed of evolution}, i.e., the inverse of the minimum time it takes the system to evolve to another (partially) distinguishable state, is a measure of asymmetry relative to the time translations generated by the system Hamiltonian. Furthermore, the celebrated  Mandelstam-Tamm  and  Margolus-Levitin speed limits can be interpreted as upper bounds on this measure of asymmetry by functions which are themselves measures of asymmetry in the special case of pure states.  Using measures of asymmetry that are not restricted to pure states, such as the Wigner-Yanase skew information, we obtain extensions of  the Mandelstam-Tamm bound which are significantly tighter in the case of mixed states. We also clarify some confusions  in the literature about coherence and asymmetry, and show that measures of coherence are a proper subset of measures of asymmetry.
\end{abstract}
\maketitle

\newpage

\section{Introduction}

Quantum Speed Limits (QSL)  are fundamental bounds on the minimum time that it 
 takes  a quantum system to evolve  to a different state. QSLs have many applications, for instance, in quantum control, quantum computation, communication, and metrology. The most famous examples are the Mandelstam-Tamm  \cite{QSL_MT} and  Margolus-Levitin bounds  \cite{Margolus:98}, which have led to numerous extensions and applications \cite{QSL_Toffoli, QSL_Fleming,QSL_Bha83,QSL_Vaidman,Pfeifer:1993fk,QSL_Pfeifer_RMP,QSL_Lloyd, QSL_Lloyd2, QSL_Zych, QSL_Chau,Deffner_uncertainty, QSL_Davidov_13,QSL_Plenio_13,QSL_Deffner_13,QSL_14_SciRep, marvian2015quantum}.  Let $\tau_\perp(\rho)$ be the minimum time that it takes, under Hamiltonian $H,$ for the state $\rho$ to evolve to a perfectly distinguishable state. Then, the Mandelstam-Tamm bound asserts that
\begin{equation}\label{QSL:MT}
\tau_\perp(\rho) \ge \frac{\pi}{2\Delta E(\rho)} \ ,
\end{equation}
where   
$\Delta E(\rho)\equiv\sqrt{\text{tr}(\rho H^2)-\text{tr}^2(\rho H)}$ is the energy uncertainty in state $\rho$ (throughout this paper we take $\hbar=1$). According to the Margolus-Levitin bound, 
\begin{equation}\label{QSL:ML}
\tau_\perp(\rho) \ge \frac{\pi}{2 \left[{E}_\text{av}(\rho)-E_\text{min}(\rho)\right]} 
\end{equation}
where ${E}_\text{av}(\rho)=\text{tr}(\rho H)$ is the average energy of state $\rho$ and  $E_\text{min}(\rho)$ is the minimum energy level of Hamiltonian $H$ in which state $\rho$ has a nonzero component \cite{Margolus:98}.  Several generalizations of these bounds have been found (See, e.g., \cite{QSL_Toffoli, QSL_Fleming,QSL_Bha83,QSL_Vaidman,Pfeifer:1993fk,QSL_Pfeifer_RMP,QSL_Lloyd, QSL_Lloyd2, QSL_Zych, QSL_Chau,Deffner_uncertainty, QSL_Davidov_13,QSL_Plenio_13,QSL_Deffner_13,QSL_14_SciRep, marvian2015quantum, mondal2014quantum}). In particular, Giovannetti et al. \cite{QSL_Lloyd}  generalized these bounds by finding the lower bounds on the minimum time it takes for the system to evolve to a state having fidelity $\delta$ with the initial state. These lower bounds are basically the same as the original Mandelstam-Tamm and  Margolus-Levitin bounds, Eq.(\ref{QSL:MT}) and Eq.(\ref{QSL:ML}) up to a multiplicative factor that only depends on the fidelity $\delta$.   

Although both the Mandelstam-Tamm and the Margolus-Levitin QSL bounds are attainable for pure states, for a general mixed state they can be rather loose. For instance, if the state is \emph{incoherent} in the energy eigenbasis, i.e., diagonal in this basis, then it does not evolve. So, $\tau_\perp$ is infinite, and therefore $\tau^{-1}_\perp$, the speed of evolution, is zero. However,  in this case  the lower bounds on $\tau_\perp$  implied by Mandelstam-Tamm and  Margolus-Levitin QSLs can be arbitrarily small. In other words, in the case of states that are incoherent in the energy eigenbasis, the quantities $\Delta E$ and ${E}_\text{av}(\rho)-E_\text{min}(\rho)$ do not contain any information about the speed of evolution. 
All of this suggests that we might be able to find tighter quantum speed limits by quantifying the coherence of states relative to the energy eigenbasis. 

In recent years, two different approaches for quantifying the coherence of states have been proposed in the literature. The first approach defines coherence as asymmetry relative to a translational symmetry, such as time-translations or phase-shifts\cite{Modes, Noether, marvian2016quantify, QRF_BRS_07,lostaglio2015quantum}, while the second approach, proposed by Baumgratz et. al. \cite{Coh_Plenio} defines coherence as a resource which cannot be generated under \emph{incoherent operations}. (See Sec.(\ref{Sec:coh}) for a short review).  

In this paper, we will show that formalizing the notion of \emph{speed of evolution} naturally leads us to the notion of coherence as \emph{asymmetry relative to time translations}. Indeed, we will show that any notion of speed of evolution of a closed system is a measure of asymmetry relative to time translations generated by the system Hamiltonian.  Interestingly, it turns out that  Mandelstam-Tamm  and  Margolus-Levitin QSLs can both be interpreted as upper bounds on this measure of asymmetry by functions which are themselves measures of asymmetry in the case of pure states.  The variance of energy, $\Delta E^2$, for instance, is a measure of time-translation asymmetry on pure states.  A genuine measure of asymmetry, however, is one that applies to all states, not just pure states. Several of these have been recently constructed.  The Wigner-Yanase skew information is an example \cite{Marvian_thesis, Noether}.  We here show that by considering genuine measures of asymmetry in the case of time-translations,  
we can obtain extensions of the  Mandelstam-Tamm bound which are significantly tighter in the case of mixed states.  Note that throughout this paper we only consider time-independent Hamiltonians. 

We start with a short review of the resource theory of asymmetry and a discussion of the different approaches for quantifying coherence. We also clarify some confusions in the literature about concepts of asymmetry and coherence (See \cite{marvian2016quantify} for further discussions).

\section{Quantifying coherence}\label{Sec:coh}

In recent years, two slightly different approaches have been proposed for treating 
coherence as a resource.

The first approach defines coherence as \emph{asymmetry relative to a group of translations}, such as phase shifts or time translations \cite{Modes, Noether, lostaglio2015quantum, lostaglio2015description}. As we will see in the following, this is the notion of coherence which naturally appears in the context of QSLs. This notion of coherence has also been extensively used in the context of quantum thermodynamics, (See  e.g. \cite{lostaglio2015quantum, lostaglio2015description})  quantum optics and reference frames (See e.g. \cite{QRF_BRS_07, MS11, GMS09, Marvian_thesis}), and quantum metrology (See \cite{marvian2016quantify} for further discussions).   Indeed, the study of coherence as a resource has been one of the primary motivations for the developments of the theory of quantum reference frames and the resource theory of asymmetry \cite{Modes, QRF_BRS_07}.  In all these physical examples,  there is  a fundamental or an effective translational symmetry in the problem, or there is an additive conserved  observable, such as energy, momentum, angular momentum or total photon number. For instance, as is discussed in detail in a recent paper by Lostaglio et. al. \cite{lostaglio2015quantum}, this notion of coherence naturally shows up in the context of quantum thermodynamics, where the only \emph{free unitaries} are the energy-conserving ones. In Sec.(\ref{Sec:coh}) we briefly review the resource theory of asymmetry for the special case of translational symmetries.

On the other hand,  Baumgratz et al. have proposed a different approach for quantifying coherence \cite{Coh_Plenio}.  Given some preferred basis, it is natural to define the set of \emph{incoherent states} as those that are diagonal in this basis.  Baumgratz et al. define the set of {\em incoherent operations} as those quantum operations for which there exists a Kraus decomposition $\mathcal{E}(\cdot)=\sum_\mu K_\mu(\cdot) K^\dag_\mu $ such that for each Kraus operator $K_\mu$ and any incoherent state $\rho$,  $K_\mu\rho K^\dag_\mu/\text{tr}(K_\mu\rho K^\dag_\mu)$  is also an incoherent state.  In this approach, coherence is defined as the resource relative to the set of incoherent operations.
Specifically, according to this proposal, measures of coherence are functions over the states that are non-increasing under incoherent operations.

In the rest of this section, we give a short review of the resource theory of asymmetry for the special case of time-translations and we study the relation between the notion of coherence as translational asymmetry and the notion proposed by Baumgratz et al.

\subsection{Coherence as asymmetry relative to translations}\label{Sec:coh}

The resource theory of asymmetry is a framework for quantifying and classifying asymmetry of states and operations \cite{gour2008resource, Noether, GMS09, MS11, MS_Short, Marvian_thesis} (See \cite{brandao2015reversible, coecke2014mathematical, fritz2015mathematical} for a general discussion of resource theories). In the special case where the symmetry group is the set of translations generated by a fixed observable, asymmetry can be interpreted as coherence relative to the  eigenbasis of this observable, and the resource theory of asymmetry provides a framework to study this notion of coherence \cite{Modes, QRF_BRS_07, Noether, lostaglio2015quantum, lostaglio2015description}.

To characterize coherence relative to the eigenbasis of an observable, say a  time-independent  Hamiltonian  $H$, we consider the one-parameter group of  unitaries generated by this observable,  the set of time translations $\{e^{-i H t}: t\in\mathbb{R}\}$.    If the eigenvalues of the generator $H$ are all separated from each other  by a constant times integers, then the group of translation is isomorphic to $\text{U(1)}$, the group of phases.\footnote{In this case sometimes asymmetry is called \emph{$\text{U(1)}$-asymmetry.}}  This happens, for instance, in the case of total photon number, which generates phase shifts, or equivalently, in the case of the Hamiltonian for a harmonic oscillator.

In this resource theory, \emph{free states}  are defined as the states with no asymmetry, i.e., states which are invariant under all time translations,
\beq
 e^{-i H t} \rho_\TI e^{i H t}=\rho_\TI\ , \ \ \   \forall t\in\mathbb{R} \ .
\eeq 
Here the subscript TI stands for \emph{Translationally Invariant}. Clearly these consist of all and only the states which are diagonal in the Hamiltonian eigenbasis, i.e.,
\beq\label{TI_op}
\rho_\TI\in \mathcal{I}_H \  \Longleftrightarrow\     e^{-i H t} \rho_\TI e^{i H t}=\rho_\TI\ : \forall t\in\mathbb{R} \ ,
\eeq 
where  $\mathcal{I}_H$ is the set of \emph{incoherent states} in the energy eiegnbasis. 
In other words, \emph{incoherence relative to the Hamiltonian eigenbasis is equivalent to invariance under time-translations}.

Similarly, a trace-preserving completely positive map, i.e., a \emph{quantum operation}, is a \emph{free operation} in the resource theory of time-translation asymmetry if it is invariant under all time translations, 
that is,  if it satisfies 
 \begin{align}
e^{-i H t}\mathcal{E}_{\text{TI}}(\rho) e^{i H t}&=\mathcal{E}_\text{TI}\left(e^{-i H t}\rho e^{i H t}\right)\label{TI_op}\ , \ \ \   \forall t\in\mathbb{R}  ,
\end{align}
for any input state $\rho$. 
Translationally invariant quantum operations are termed \emph{TI operations} in this paper.  As we discuss in Sec.\ref{Levitin}, any TI operation can be implemented by  applying an energy-conserving unitary on the system and an \emph{environment}, which is initially in an incoherent state (See also \cite{yang2015optimal}). In other words, TI 
quantum operations consist of all and only those operations which can be implemented on the system under the restriction to time-invariant resources.

Clearly  TI quantum operations  cannot create coherence in the energy eigenbasis,
\beq\label{TI_op}
\rho_\TI\in \mathcal{I}_H \  \Longrightarrow\  \mathcal{E}_\TI(\rho_\TI)\in \mathcal{I}_H \ .
\eeq 
Motivated by this observation,  in this approach, coherence relative to the eigenbasis of $H$ is defined as the resource under TI operations, and therefore is quantified using \emph{measures of asymmetry} for the group of translations generated by $H$; i.e., functions satisfying the following definition
\begin{definition}\label{def1}
A function $f$ from states to real numbers is a measure of asymmetry with respect to translations generated by a given observable $H$, if it satisfies \\
\noindent (i) For any TI quantum operation $\mathcal{E}_{\text{TI}}$, and any state $\rho$ it holds that $f(\mathcal{E}_{\text{TI}}(\rho))\le f(\rho)$.\\
\noindent (ii) For any incoherent state $ \rho_\text{TI}\in \mathcal{I}_H$, it holds that $f(\rho_\text{TI})=0$.
\end{definition}
 Note that the second condition is simply a convention which fixes the value of function $f$ on incoherent states, and guarantees that it is a non-negative function of states. This is true because for any incoherent state there is  a TI operation which maps its input to that incoherent state\footnote{For instance,  the quantum operation which discards the input state and prepares the desired incoherent state.}, 
and so any function which satisfies condition (i) should take the same value on all incoherent states, and this should be the minimum value that function takes on all states. Therefore, by shifting the function by a constant, one can always make sure that it satisfies condition (ii) as well and is non-negative.

Also, note that for closed-system dynamics under Hamiltonian $H$, any measure of asymmetry (relative to time translation) remains constant, i.e.,
\begin{equation}\label{inv4}
\forall t\in \mathbb{R}:\ \ \  f(\rho)=f(e^{-i H t}\rho\ e^{i H t})\ ,
\end{equation}
for any state $\rho$. 
This follows from the fact that at any time $t$ the map $\rho\rightarrow e^{-i H t} \rho e^{i H t}$ is a TI quantum operation, and it can  be inverted by another TI quantum operation, namely $\rho\rightarrow e^{i H t} \rho e^{-i H t}$. 

In recent years, many examples of measures of asymmetry have been studied in the literature (See e.g. \cite{gour2008resource, Modes, GMS09, Marvian_thesis, Noether, skotiniotis2012alignment, narasimhachar2014phase, toloui2011constructing, vac2008, girolami2015witnessing, toloui2012simulating}). In particular, \cite{Noether, Marvian_thesis} propose a general recipe for constructing measures of asymmetry. Using this recipe, for instance,  it is shown that the function  
\begin{align}
F_H(\rho)&\equiv \| [H,\rho]\|_1
\label{F_H}
\end{align}
is a \emph{faithful} measure of asymmetry  \cite{Noether, Marvian_thesis}, 
 where faithfulness means that  it vanishes if and only if the state is incoherent (In this paper $\|\cdot\|_1$ denotes the $l_1$-norm, i.e., the sum of singular values of the operator).  Later, we will present some interesting properties of this measure of asymmetry and show that it is indeed relevant in the context of quantum speed limits.

\subsection{Relation between the two approaches}


In this section, we study the relation between understanding coherence as asymmetry relative to a  group of translations and understanding coherence in the manner proposed by Baumgratz et al. \cite{Coh_Plenio} and we briefly discuss the applications of the first approach (See \cite{marvian2016quantify} for further discussion).
  Notice that although in this paper we often assume that the generator of the translations is the system's Hamiltonian, the following discussion holds for any other observable, such as photon number or linear momentum or angular momentum.

According to Eq.(\ref{TI_op}) under TI quantum operations any state which is incoherent  in the eigenbasis of the generator of translations   evolves to a state which is still incoherent in this basis. Moreover, as is shown in the appendix (See also \cite{yang2015optimal}),
\begin{proposition} \label{prop_incoh}
For any given observable $H$ (in particular, the
Hamiltonian), all TI operations are incoherent operations (in the
sense defined by Baumgratz et al. \cite{Coh_Plenio}, relative to the eigenspaces of
$H$).  Therefore, any measure of coherence in the sense of Baumgratz et al. \cite{Coh_Plenio}, i.e., a function that is non-increasing under incoherent operations, is also a measure of translational asymmetry. 
\end{proposition} 
As a matter of fact, it turns out that almost all measures of coherence which have been found recently,  have been previously studied  in  the resource theory of asymmetry. For instance, the function called \emph{relative entropy of coherence} by Baumgratz et. al. \cite{Coh_Plenio}  has been extensively studied as a measure of asymmetry under the names of \emph{G-asymmetry} and \emph{relative entropy of asymmetry} \cite{vac2008, GMS09, brandao2015reversible, toloui2012simulating}, and it has been generalized to a family of measures of asymmetry, called \emph{Holevo asymmetry measures} \cite{Noether, Marvian_thesis} (See also  \cite{Modes, Marvian_thesis} for measures of asymmetry based on $l_1$ norm).

On the other hand,  there are incoherent operations (unitaries) which are not translationally invariant. For instance, consider permutations of the eigenvectors of the generator $H$, i.e.,  unitaries in the form
\beq
U_\sigma= \sum_i |\sigma(i)\rangle\langle i|\ ,
\label{U_sigma}
\eeq
where $\sigma$ is an arbitrary permutation of the elements of the eigenbasis of $H$.
It can be easily seen that, while all these unitaries are incoherent operations, in general they are not TI operations. Thus, TI quantum operations are a proper subset of incoherent operations (See Fig.~\ref{Fig}).

\begin{figure} [h]
\begin{center}
\includegraphics[scale=.38]{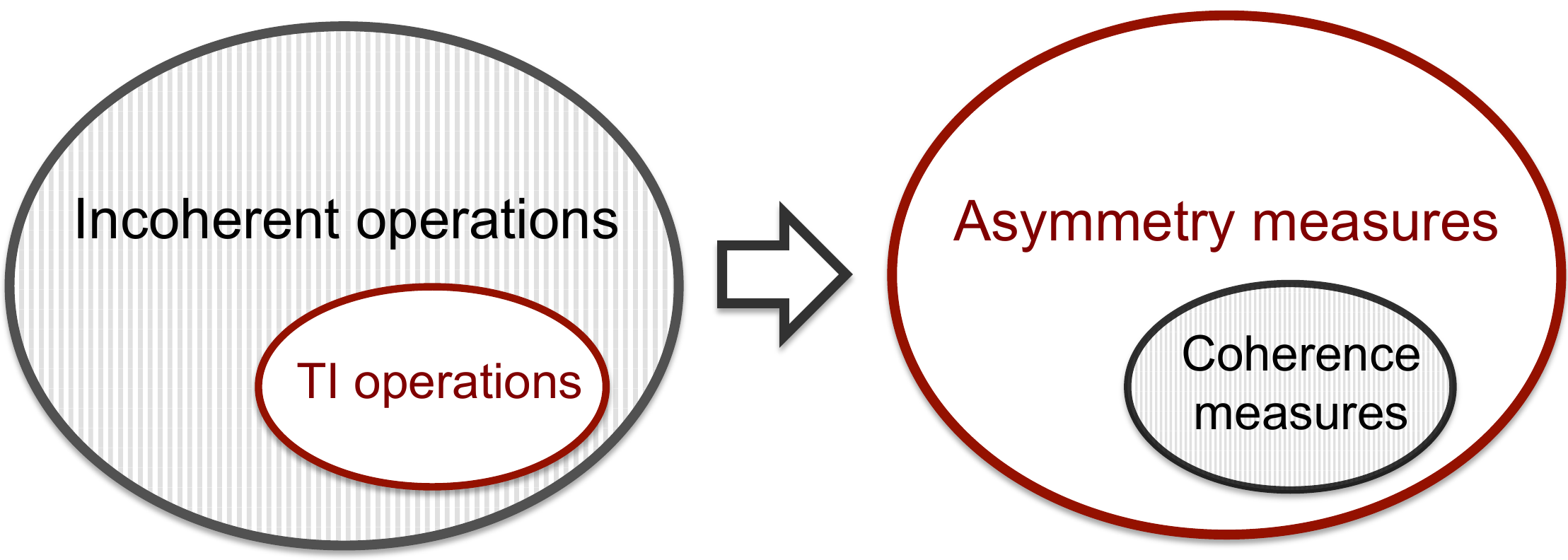}
\caption{Translationally invariant operations are a proper subset of incoherent operations. Consequently, measures of coherence (in the sense defined by Baumgratz et al. \cite{Coh_Plenio}), i.e. functions which are non-increasing under incoherent operations, are a proper subset of measures of asymmetry. }   \label{Fig}
\end{center}
\end{figure}    

Moreover, it turns out that there are measures of translational  asymmetry which are not measures of coherence in the sense of Baumgratz et al., i.e., they can increase under incoherent operations.
 In particular, any function $f_H$ which is a measure of coherence according to the definition of \cite{Coh_Plenio} should satisfy 
 \beq\label{cond}
f_H(\rho)=f_{H}\left(U_\sigma \rho U^\dag_\sigma \right) \ ,
\eeq
for any unitary $U_\sigma$ of the form (\ref{U_sigma}), and any state $\rho$, and,
remarkably,  condition (\ref{cond})  is not necessarily satisfied by all measures of asymmetry. In particular, it is not satisfied by measures of asymmetry which are relevant in the context of QSLs, such as the function $F_H$ introduced above, or the Wigner-Yanase skew information (See section \ref{skew}).

This fact is  related to an important distinction between the two approaches for quantifying coherence: unlike the first approach  based on the notion of asymmetry, in the approach of Baumgratz et. al. \cite{Coh_Plenio} the eigenvalues of the observable which defines the preferred basis relative to which coherence is defined  are irrelevant. 

However, in many physical applications where the notion of coherence is important,  the eigenvalues of the observable which defines the preferred basis play an important role.

 As a simple example, consider the problem of phase estimation, where light in a particular mode is sent through an optical element that generates an unknown phase shift $e^{i\theta}$, and the goal is to estimate this phase shift. In this context, coherence between states with different numbers of photons is an essential resource: incoherent states are useless for phase estimation. Let us now consider the two states  $|0\rangle+|1\rangle$ and $|0\rangle+|N\rangle$ where $N > 1$, which both contain coherence.  From the point of view of the resource theory of coherence proposed by Baumgratz et. al. \cite{Coh_Plenio},  these states are equivalent resources, because they can be interconverted to each other by  incoherent operations of the form (\ref{U_sigma}), and therefore any measure of coherence takes the same value on them. In spite of this fact, in the above phase estimation task, the information one can obtain about the unknown phase $e^{i\theta}$  will be  different for these two  states: in one case  the relative phase shift of the two terms in the state is $e^{i\theta}$, and in the other it is $e^{i N \theta}$. Thus,  in this context, these  states are not equivalent resources, and therefore their usefulness for the task of phase estimation
cannot be quantified using the measures of coherence proposed by Baumgraz et al. 
On the other hand, measures of translational asymmetry , for instance,  $F_H$ in Eq. ~(\ref{F_H}), can capture the difference between these two states. Moreover, it turns out that \emph{any function which quantifies the performance of states for the task of phase estimation is automatically a measure of asymmetry relative to phase shifts} \cite{marvian2016quantify}. 


Not only in quantum metrology, but in physical contexts such as quantum thermodynamics and quantum reference frames and QSLs, which will be studied in this paper, the particular eigenvalues of the observable which defines the preferred basis are also significant. In such cases,  the measures of coherence (in the sense of \cite{Coh_Plenio}) provide a very limited characterization of coherence of states; to find a complete  characterization one needs to use measures of translational asymmetry, which form a larger set of functions.  In this paper, we study the case of QSLs, where the notion of coherence naturally shows up, and we show that, while measures of coherence (in the sense of \cite{Coh_Plenio}) cannot capture any information about the speed of evolution, any natural notion of speed of evolution is automatically a measure of asymmetry (See also \cite{marvian2016quantify} for further discussions on other physical examples).

\subsection{Is Wigner-Yanase skew information a measure of coherence?}\label{skew}

In 1963, Wigner and Yanase \cite{wigner1963information}  introduced the function 
\bes\label{S_H}
\begin{align}
S_H(\rho) \equiv \frac{1}{2}\|[H, \sqrt\rho ]\|_2^2= -\frac{1}{2}\text{tr}\left( [H, \sqrt\rho]^2 \right)\\ = \text{tr}\left( H^2 \rho \right)-\text{tr}(\sqrt{\rho} H\sqrt{\rho} H)\ ,
\end{align}
\ees
now called the \emph{Wigner-Yanase skew information}, and proved that it had certain interesting properties, 
such as convexity and additivity.\footnote{The $l_2$ norm is defined by $\|X\|_2=\sqrt{\text{tr} (X^\dagger X)}$}.  Notice that if $\rho$ is a pure state, then $\rho=\sqrt{\rho}$ and $S_H(\rho)$ reduces to the variance of $H$. Later, Dyson generalized this to the function $-\text{tr}([\rho^s,H][\rho^{1-s},H])$ for $0\le s\le 1$, which is sometimes called the Dyson-Wigner-Yanase skew information, and Lieb famously proved the convexity of this function for $0<s<1$ \cite{lieb1973convex}.

Wigner and Yanase proposed $S_H(\rho)$ as  a measure of information  and, equivalently,   $-S_H(\rho)$ as a measure of entropy for  the situations where the observable $H$ is an additive conserved quantity such as charge or components of linear or angular momenta \cite{wigner1963information}. 
Alternatively,  $S_H(\rho)$ is sometimes regarded as a measure of the \emph{non-commutativity} of the state $\rho$ and the observable $H$ (See, e.g., \cite{wehrl1978general}).

In \cite{Noether, Marvian_thesis} a new interpretation of this function  was unveiled. It was shown that Wigner-Yanase  skew information is a measure of asymmetry, 
 and therefore quantifies symmetry-breaking  relative to translations generated by $H$.  In fact, even more generally,  in \cite{Noether, Marvian_thesis}  it was shown that the  Dyson-Wigner-Yanase skew information
  is also a measure  of asymmetry for  $s\in (0,1)\cup (1,2]$.

Recently, Girolami  \cite{Girolami} proposed  an experimental method for measuring the Wigner-Yanase skew information, and argued that this function is a good candidate for quantifying coherence. Furthermore, he claimed  that this function is  a measure of coherence according to the definition of Baumgratz et. al. \cite{Coh_Plenio},  that is, he claimed that it is non-increasing under incoherent operations. However, the latter claim is incorrect. This can be  seen, for instance, by noting that in the case of pure states this function is equal to the variance of the observable $H$, but variance obviously is  not invariant under operations (\ref{U_sigma}), i.e., it violates Eq.(\ref{cond}) \footnote{The increase of skew information under incoherent operations is also observed in \cite{du2015wigner}, by looking through an explicit example.}.

To summarize, the Wigner-Yanase skew information $S_H$  is a measure of asymmetry relative to the group of translations generated by the observable $H$, that is,  $\{e^{-iH t}: t\in\mathbb{R}\}$.  Furthermore, as we discussed before, any such measure of asymmetry can be used to quantify the coherence of a state relative to the eigenbasis of $H$ and this quantification of coherence has nontrivial applications, for instance, in the context of quantum metrology, quantum reference frames,  and quantum speed limits (as will be shown in this paper).  However, this function does not satisfy the definition of a measure of coherence according to Baumgratz et al. \cite{Coh_Plenio},  as it can increase under incoherent operations.   

In the following, we show that measures of time-translation asymmetry naturally arise in the context of quantum speed limits, and, in particular, the skew information has a very natural interpretation as \emph{instantaneous acceleration}. Indeed, we show that the very notion of the \emph{speed of evolution} can be interpreted as a measure of time-translation asymmetry.

\section{Speed of evolution}

The standard quantum speed limits in Eq.(\ref{QSL:MT}) and Eq.(\ref{QSL:ML}) are lower bounds on $\tau_\perp(\rho)$, the minimum time it takes, under Hamiltonian $H$, for state $\rho$ to evolve to a perfectly distinguishable  state $\rho(t)\equiv e^{-i H t}\rho\ e^{i H t}$. Consequently,  the function $\tau_\perp^{-1}$ can be interpreted as the \emph{(average) speed of evolution}. 
 It is useful to consider generalizations of the function $\tau_\perp$ to cases where the states $\rho$ and $\rho(t)$ are only partially distinguishable.
\subsection{Measures of distinguishability} 

Quantum information theory provides different tools for quantifying the distinguishability of a pair of states. In particular, we are interested in functions from pairs of states to the real numbers, with the following three properties:\\
\textbf{(i)}  monotone under information processing, i.e., satisfying Eq.(\ref{info_ineq}),\\
\textbf{(ii)}  vanishing  when the two input states are the same, i.e., satisfying Eq.(\ref{info_ineq2}), and \\
\textbf{(iii)} jointly quasi-convex, i.e., satisfying Eq.(\ref{quasi}).\\
In the following we provide the formal definition and discuss the significance of each of these properties. We also review some examples of functions satisfying all of these properties.

The most important property of measures of distinguishability is monotonicity under \emph{information processing}. This means that  for any quantum operation $\mathcal{E}$ and for any pair of states $\sigma_1$ and $\sigma_2$, a measure of distinguishability $D$ should satisfy the \emph{information processing inequality}, 
\begin{equation}\label{info_ineq}
D\left(\mathcal{E}(\sigma_1),\mathcal{E}(\sigma_2)\right)\le D(\sigma_1,\sigma_2)\ .
\end{equation}
Note that set of quantum operations, i.e., the completely positive trace-preserving maps, include all and only the physical transformations that one can implement on a quantum system without any prior 
 information about its initial state.
Thus, satisfying this bound is the minimum requirement that any measure of distinguishability should satisfy.

In this paper we also assume that measures of distinguishability vanish when the two input states are the same
\begin{equation}\label{info_ineq2}
D\left(\sigma,\sigma\right)=0\ .
\end{equation}
Notice that this assumption is basically just  a convention: any function satisfying the information processing inequality, Eq.~(\ref{info_ineq}), can be shifted by a constant to satisfy Eq.(\ref{info_ineq2}) as well.\footnote{For any function $D$ which satisfies the information processing inequality, the value of $D(\sigma,\sigma)$ is independent of the state $\sigma$, because for any pair of states $\sigma_1$ and $\sigma_2$ there is a quantum operation that maps one to the other, and so  according to Eq.(\ref{info_ineq}), $D(\sigma_1,\sigma_1)=D(\sigma_2,\sigma_2)$. } 
It follows from Eqs.~(\ref{info_ineq}) and (\ref{info_ineq2})  that the function $D$ is non-negative.
Finally note that  $D,$ contrary to a true distance measure, 
does not have to be symmetric in its arguments, i.e., in general   $ D\left(\sigma_1,\sigma_2\right)\neq D\left(\sigma_2,\sigma_1\right)$.

In this paper we are are going to focus on measures of distinguishability  $D$ which are \emph{jointly quasi-convex}, meaning that for all $0\le p\le 1$ and any two pairs of states, $(\rho_1,\sigma_1)$ and $(\rho_2,\sigma_2)$, 
$D$ satisfies the following  inequality
\begin{align}\label{quasi}
D&\Big(p\rho_1+(1-p)\rho_2\ , p\sigma_1+(1-p)\sigma_2 \Big)\nonumber\\ &\ \ \ \ \ \ \ \ \ \ \ \ \le \text{max}\big\{D\left(\rho_1, \sigma_1 \right) ,  D\left(\rho_2, \sigma_2 \right) \big\}\ .
\end{align}
This inequality is a  weakening
 of \emph{joint convexity}, which is defined as 
\begin{align}
D&\Big(p\rho_1+(1-p)\rho_2\ ,\ p\sigma_1+(1-p)\sigma_2 \Big)\nonumber\\ &\ \ \ \ \ \ \ \ \ \ \ \ \le p D\left(\rho_1, \sigma_1 \right)+(1-p) D\left(\rho_2, \sigma_2 \right)\ .
\end{align}

Joint quasi-convexity of a measure of distinguishability ensures that 
the pair of states obtained by taking the mixture of a collection of pairs of states (where the mixing weights are the same for each element of the pair) are never more distinguishable than the most distinguishable pair in the collection. 
Joint convexity, on the other hand, asserts that the pair of states obtained by mixing a collection of pairs has distinguishability no greater than the {\em weighted average} of the distinguishabilities of the pairs in the collection.  
Clearly, joint convexity is a much stronger requirement than joint quasi-convexity.
The intuitive notion that a measure of distinguishability should be non-increasing under mixing, which is often given as an argument in favour of joint convexity, in fact only justifies quasi-convexity.

%
As noted earlier, defining a speed of evolution in terms of the time to reach a partially distinguishable state requires one to choose a measure of distinguishability for pairs of states.  As we will discuss in section \ref{mixing}, assuming that the measure of distinguishability satisfies joint quasi-convexity ensures that the speed of evolution of a state that is the mixture of some set of states is no greater than the fastest speed of evolution of any state in that set.  

Trace distance, relative entropy,  Renyi relative entropy and  infidelity ($1-F$ where $F$ is the fidelity) are all examples of measures of distinguishability which satisfy properties (i), (ii) and (iii). In this paper, we focus on the two particular examples of trace distance and Renyi relative  entropy. 

The trace distance between two quantum states, $\rho_1$ and $\rho_2$, is defined as  $\|\rho_1-\rho_2\|_1$, where $\|\cdot\|_1$ is the 1-norm. As Helstrom has shown~\cite{Helstrom:book}, the trace distance determines the maximum probability of successfully  determining which of the states in the pair was prepared, given a single copy, when the states have equal prior probability of having been prepared. 
It immediately follows that trace distance is non-increasing under information processing \cite{hayashi2006quantum}.  
Furthermore, the triangle inequality for the $l_1$-norm implies that the trace distance is jointly convex, and hence jointly quasi-convex.  

The second example of a measure of distinguishability that we use in this paper is the \emph{Renyi quantum  relative entropy}, introduced by Petz as one of the quantum generalizations of  (classical) Renyi relative entropy \cite{petz1986quasi, hayashi2006quantum}.
\footnote{Note that this definition is different from the ``sandwiched'' Renyi relative entropy \cite{muller2013quantum,wilde2014strong}.}
 For $s\in (0,1)\cup (1,\infty)$, this function is defined as 
\begin{equation} \label{def}
D_{s}(\rho_{1},\rho_{2})\equiv \frac{1}{s-1}\log\left(\text{tr}(\rho_{1}^{s}\rho_{2}^{1-s})\right)\ ,
\end{equation}
and it satisfies both Eq.~(\ref{info_ineq}) and Eq.~(\ref{info_ineq2}) for $s\in (0,1)\cup (1,2]$ \cite{mosonyi2014convexity}. Also the function is jointly convex for  $s\in (0,1)$  \cite{lieb1973convex, petz1986quasi, mosonyi2014convexity}. In this paper,  we  focus on the case of $s=1/2$, i.e., $D_{1/2}(\rho_{1},\rho_{2})\equiv -2\log\text{tr}(\sqrt{\rho_{1}}\sqrt{\rho_{2}})$, though the idea can be generalized.

\subsection{Definition of speed of evolution}

Given any measure of distinguishability satisfying conditions (i), (ii) and (iii), we can define a notion of \emph{speed of evolution}, which generalizes the function $1/\tau_\perp$ that appears in the standard quantum speed limits. For $\epsilon>0$, let $\tau^D_\epsilon(\rho)$ be the minimum time it takes for a state $\rho$ to evolve, under Hamiltonian $H,$ to another state $\rho(t)\equiv e^{-i H t}\rho\ e^{i H t}$ which is at least $\epsilon$-distinguishable from state $\rho$ relative to $D$, i.e. $D\left(\rho, \rho(t)\right)\ge \epsilon$.  If this never happens for $t>0$, we define $\tau^D_\epsilon(\rho)$  to be infinity. So, to summarize
\begin{equation}\label{def:t}
 \tau^D_\epsilon(\rho)\equiv\begin{cases}
    \infty\ ,  \ \ \ \ \ \ \ \ \ \ \ \ \ \ \ \ \text{if}\ \ \forall t\in\mathbb{R}^+:  \ \  D\left(\rho, \rho(t)\right)< \epsilon  \\
    \min\{t: t\in\mathbb{R}^+\ , \ D\left(\rho, \rho(t)\right)\ge \epsilon \}\ , \ \  \text{otherwise},
  \end{cases}
\end{equation}
or equivalently, $\tau^D_\epsilon(\rho)\equiv \text{sup}\{t: D\left(\rho, \rho(t')\right)< \epsilon, \forall t'\in (0,t)   \}$.
Therefore, 
for any $\epsilon>0$ and any measure of distinguishability $D$ which satisfies conditions (i), (ii) and (iii),  the function $1/{\tau^D_\epsilon(\rho)}$ (or $\epsilon/{\tau^D_\epsilon(\rho)}$) defines a natural notion of (average) \emph{speed of evolution}. 

A simple example of a measure of distinguishability that satisfies conditions (i), (ii), and (iii)  is the function $D_\perp(\rho,\sigma)$, defined to be one if and only if the two states $\rho$ and $\sigma$ are perfectly distinguishable (which requires them to have orthogonal supports) and zero otherwise. 
 Starting from this measure of distinguishability, and using the definition (\ref{def:t}) for $\epsilon=1$, one obtains  the function $\tau_\perp$ that appears in the Mandelstam-Tamm and  Margolus-Levitin bounds Eqs.(\ref{QSL:MT} and \ref{QSL:ML}). 
Thus, the corresponding speed of evolution, $\tau^{-1}_\perp$, satisfies the  above definition.  Later, we will consider two  other examples of functions $\tau^D_\epsilon$ which are obtained based on the trace distance and Renyi   relative entropy, as measures of distinguishability.

\subsection{Speed of evolution is a measure of asymmetry}

Next, we present our first result on the connection between quantum speed limits and measures of asymmetry. Recall that any incoherent state in the Hamiltonian eigenbasis, i.e., any member of $\mathcal{I}_H$, commutes with the Hamiltonian, and so it remains invariant under the evolution generated by this Hamiltonian. Therefore, relative to any measure of distinguishability $D$, its corresponding speed $1/\tau^D_\epsilon$ is zero. Intuitively, one may expect that having a higher speed of evolution  corresponds to  being {\em less} invariant under time-translation, which is to say having  a  higher amount of asymmetry relative to time-translation, or equivalently, a higher amount of coherence relative to the eigenbasis of $H$. The following theorem confirms this intuition. 
\begin{theorem}\label{prop1}
For any measure of distinguishability $D$ that satisfies the information processing inequality, Eq.(\ref{info_ineq}) and vanishes when the two states are the same, Eq.(\ref{info_ineq2}), and for any $\epsilon>0$, the function $1/\tau^D_\epsilon$ is a measure of asymmetry relative to the time-translations (generated by the system Hamiltonian $H$). 
\end{theorem}
\begin{proof} First consider the case where  $\tau^D_\epsilon(\rho)<\infty$ for the given state $\rho$. In this case there is a finite time $t$ at which $\rho$ and $\rho(t)$ are at least $\epsilon$-distinguishable relative to $D$, i.e.,  $D(\rho,\rho(t))\ge \epsilon\ $.
Let  $\mathcal{E}_\text{TI}$ be an arbitrary TI quantum operation.
Then, it holds that 
\begin{align}\label{Eq1}
&\tau^D_\epsilon\left(\mathcal{E}_\text{TI}(\rho)\right)\nonumber\\ &= \min\{t: t\ge 0\ , D\left(\mathcal{E}_\text{TI}(\rho), e^{-i H t}\mathcal{E}_\text{TI}(\rho)\ e^{i H t}\right)\ge \epsilon \}\nonumber \\ &= 
\min\{t:  t\ge 0\ , D\left(\mathcal{E}_\text{TI}(\rho), \mathcal{E}_\text{TI}(e^{-i H t}\rho\ e^{i H t})\ \right)\ge \epsilon \} \nonumber \\ &\ge 
\min\{t:  t\ge 0\ , D\left(\rho, \rho(t)\ \right)\ge \epsilon \}\nonumber \\ &= 
\tau^D_\epsilon\left(\rho\right) ,
\end{align}
where to get the third line, we have used the time-translational symmetry of $\mathcal{E}_\text{TI}$, i.e., Eq.(\ref{TI_op}), and to get the fourth line we have used the  information processing inequality (\ref{info_ineq}), which implies that for any time $t$,
\beq
D\left(\rho, \rho(t)\right)\ge D\left(\mathcal{E}_\text{TI}(\rho), \mathcal{E}_\text{TI}(\rho(t))\ \right)\ ,
\eeq
and so the minimum $t$ in the third line  should be larger than or equal to the minimum $t$ in the fourth line. Using a similar argument one can easily see that if $\tau^D_\epsilon(\rho)=\infty$, i.e., if the distinguishability of $\rho$ and $\rho(t)$ is always less than $\epsilon$, then the distinguishability of $ \mathcal{E}_\text{TI}(\rho)$ and $e^{- i H t} \mathcal{E}_\text{TI}(\rho) e^{i H t}$ is also always less than $\epsilon$, and so $\tau^D_\epsilon( \mathcal{E}_\text{TI}(\rho) )=\infty$. So, in general, we find that for any TI quantum operation $ \mathcal{E}_\text{TI}$, it holds that
$\tau^D_\epsilon(\rho )\le \tau^D_\epsilon( \mathcal{E}_\text{TI}(\rho) ),$
and hence
\beq
\frac{1}{\tau^D_\epsilon(\rho)}\ge  \frac{1}{\tau^D_\epsilon(\mathcal{E}_\text{TI}(\rho))}\ .
\eeq
Therefore, the function $1/\tau^D_\epsilon$ satisfies condition (i) in the definition of a  measure of translational asymmetry (Definition \ref{def1}). Finally, note that   any incoherent state $\rho_\text{TI}$ is invariant under time evolution, and so for any $\epsilon>0$, $\tau^D_\epsilon\left(\rho\right)=\infty$ which implies that the speed $1/\tau^{D}_\epsilon(\rho)=0$, and therefore that $1/\tau^D_\epsilon$ satisfies condition (ii) in Definition \ref{def1} as well. This completes the proof of the theorem.\end{proof}\\
This theorem shows clearly  why the notion of coherence as asymmetry relative to time translation naturally appears in the context of quantum speed limits: because the very notion of speed itself is a measure of asymmetry relative to time-translation.  Note that the speed of evolution can, however, increase (unboundedly) under what Ref.~\cite{Coh_Plenio} termed \emph{incoherent operations},
 and so the notion of coherence studied by Baumgratz et. al.  \cite{Coh_Plenio} does not characterize the  speed of evolution.


Theorem \ref{prop1} leads to a useful framework for understanding and generalizing quantum speed limits. According to this theorem, any  function which can quantify the notion of speed of evolution  should be  non-increasing under TI quantum operations, and hence should be a measure of asymmetry relative to time-translation.
In other words, a quantity which can be increased under TI quantum operations  is not a natural candidate for quantifying the speed of evolution. This suggests that to find tighter quantum speed limits, one should try to find inequalities which can be expressed  in terms of asymmetry measures. Note that Eq.(\ref{inv4}) guarantees that any such function remains constant during the evolution.

Perhaps surprisingly, it turns out that  both the standard  Mandelstam-Tamm and Margolus-Levitin bounds  satisfy this property for pure states. For a pure state $\psi$, these bounds provide an upper bound on the speed of evolution as 
\beq \label{ineq4}
\tau^{-1}_\perp(\psi) \le \frac{2\Delta E(\psi)}{\pi}\ ,  \frac{2\left[{E}_\text{av}(\psi)-E_\text{min}(\psi)\right]}{\pi} . 
\eeq
As we   show in Sec.(\ref{Sec:Gen}) and Sec.(\ref{Levitin}), the right-hand side of both of these bounds are also non-increasing under TI quantum operations.
 Therefore in the case of pure states, inequalities (\ref{ineq4}) can be interpreted as upper bounds on a measure of asymmetry, namely, $\tau^{-1}_\perp(\psi) $, by two other measures of asymmetry, namely, $\Delta E(\psi)$ and ${E}_\text{av}(\psi)-E_\text{min}(\psi)$.  
However, it can be easily shown that for mixed states these functions can, in general, increase under TI quantum operations, and hence they are not  measures of asymmetry. For instance, the operation which maps any quantum state to the completely mixed state is clearly TI. However, for the completely mixed state the  variance of energy can  be arbitrarily large. In this case both Mandelstam-Tamm and Margolus-Levitin  provide very loose bounds; 
they cannot see the fact that the speed of evolution is zero. 

In Sec. (\ref{Sec:Gen}), we find generalizations of the Mandelstam-Tamm bound in which the variance $\Delta E$ is replaced by a genuine measure of time-translation asymmetry. The fact that the upper bound on the speed of evolution is a measure of asymmetry, in particular, guarantees that it vanishes for all incoherent states, including the completely mixed state.

\subsection{Mixing does not increase speed}\label{mixing}

Intuitively, one expects that  the speed of evolution of the mixture of two states is no greater than 
the fastest speed of evolution of  each of them.  
We therefore propose that  any reasonable notion of speed of state evolution should satisfy this property. In other words, if a function $f$ from states to real numbers quantifies the speed of evolution, then it should be \emph{quasi-convex}, meaning that for any $0\le p\le 1$ and for any pair of states $\rho$ and $\sigma$ it should satisfy
\beq
f\big(p\rho+(1-p) \sigma\big)  \le \text{max}\big\{f(\rho) , f(\sigma) \big\} \ .
\eeq
Quasi-convex functions are natural generalizations of convex functions, which satisfy
\beq
f\big(p\rho+(1-p) \sigma\big)  \le p f(\rho)+(1-p) f(\sigma) \ .
\eeq
Note that the monotonicity of speed under mixing  only requires  quasi-convexity of the function, and not the convexity, which is a stronger condition.   

The following proposition asserts that a speed of evolution is automatically quasi-convex if it is defined in terms of a measure of distinguishability that is jointly quasi-convex. 
\begin{proposition}\label{prop:joint}
For any jointly quasi-convex measure of distinguishability $D$, i.e., one satisfying Eq.(\ref{quasi}), and for any $\epsilon>0$, the function $1/\tau^D_\epsilon$, defined via Eq.~(\ref{def:t}), is quasi-convex, i.e., 
for any $0\le p\le 1$, and any pair of states $\rho$ and $\sigma$ it holds that
 \beq
\frac{1}{\tau_{\epsilon}^{D}\left(p\rho+(1-p)\sigma\right)} \le \text{max}\left\{\frac{1}{\tau_{\epsilon}^{D}\left(\rho\right)}\  ,\  \frac{1}{\tau_{\epsilon}^{D}\left(\sigma\right)} \right\}  \ .
\eeq
Equivalently,  the function $\tau^D_\epsilon$ is \emph{quasi-concave}, that is,
\beq
\tau_{\epsilon}^{D}\left(p\rho+(1-p)\sigma\right) \ge \text{min}\left\{\tau_{\epsilon}^{D}\left(\rho\right),\tau_{\epsilon}^{D}\left(\sigma\right) \right\} \ . 
\eeq
\end{proposition}
\begin{proof}
Let $$t_0 \equiv \tau^D_\epsilon(p\rho+(1-p)\sigma).$$  
Recall from Eq.~(\ref{def:t}) that $\tau^D_\epsilon(\nu)$ for any state $\nu$ is defined as the minimum time at which $D(\nu,\nu(t))\ge \epsilon$. It follows that
$$D\Big(p\rho+(1-p)\sigma\ ,\  p\rho(t_0)+(1-p)\sigma(t_0)\ \Big)\ge \epsilon.$$
Using the quasi-convexity of $D$, 
this implies that
\begin{align*}
&\text{max}\big\{D\left(\rho, \rho(t_0) \right) , D\left(\sigma , \sigma(t_0) \right) \big\}\\ &\ \ \ \ \ \ \ \ \ \ \ge D\big(p\rho+(1-p)\sigma\ ,\ p\rho(t_0)+(1-p)\sigma(t_0)\ \big)\\ &\ \ \ \ \ \ \ \ \ \ \ge \epsilon\ .
\end{align*}
Given the definitions of $\tau^D_\epsilon(\rho)$ and $\tau^D_\epsilon(\sigma)$, this implies that either
\beq
\tau^D_\epsilon(\rho)\le t_0\,
\eeq
 or 
\beq
\tau^D_\epsilon(\sigma) \le t_0\,
\eeq
or both,  which in turn implies that $t_0 \ge \text{min}\left\{\tau_{\epsilon}^{D}\left(\rho\right),\tau_{\epsilon}^{D}\left(\sigma\right) \right\} $.  Recalling the definition of $t_0$,  this completes the proof. 
\end{proof}\\
Since any function which can quantify speed of evolution is expected to be non-increasing under mixing, it is desirable to find quantum speed limits (upper bounds on the speed of evolution) that respect this property as well. That is,  it is desirable to find upper bounds on the speed $1/\tau_{\epsilon}^{D}$ which are also non-increasing under mixing. Note that the standard  Mandelstam-Tamm bound does not have this property, because the uncertainty $\Delta E$ in general increases under mixing. For instance, by mixing two eigenstates of energy with different energies,  each of which has vanishing $\Delta E$,  we get a state with nonzero $\Delta E$.

Based on theses ideas, in the following we present generalizations of the Mandelstam-Tamm  bound in which the speed $1/\tau^D_\epsilon(\rho)$, instead of  being bounded by $\Delta E(\rho)$, is bounded by measures of time-translation asymmetry. Moreover, the latter are quasi-convex (by virtue of being convex),  and therefore  do not increase under mixing. This  leads to tighter quantum speed limits in the case of mixed states. In particular, these generalizations would imply that for all incoherent states, the speed of evolution is zero. Also, in Sec. (\ref{Levitin}), we discuss the Margolus-Levitin bound and we show that the function  ${E}_\text{av}(\psi)-E_\text{min}(\psi)$, which shows up in this bound, is monotonic under TI quantum operations.

\section{Generalized Mandelstam-Tamm bounds}\label{Sec:Gen}
In this section, we consider two particular examples of measure of distinguishability, namely the trace distance and the Renyi  relative entropy, and we show that they lead to two different generalizations of  the Mandelstam-Tamm bound, both of which reduce to the Mandelstam-Tamm 
 bound in the special case of pure states (up to a constant of order one)  but yield tighter bounds in the case of mixed states.  These generalizations of the Mandelstam-Tamm bound have, roughly speaking, the following interpretations. Note first that $1/\tau^D_\epsilon(\rho)$ can be interpreted as an {\em average} speed of evolution relative to the distinguishability measure $D$.   Using the trace distance as our distinguishability measure, we find that this average speed of evolution is upper-bounded by the \emph{instantaneous} speed of evolution (the first derivative of  the measure).  Using the Renyi relative entropy as our measure, we find that the average speed of evolution is upper-bounded by the \emph{instantaneous acceleration} of the evolution (the second derivative of the measure).


First, consider the trace distance as the measure of distinguishability. For two states $\sigma_1$ and $\sigma_2$, it is given by  
$\|\sigma_1-\sigma_2 \|_1$, where $\|\cdot\|_1$ is the $l_1$-norm. Consider  Eq.~(\ref{def:t}), and let $\tau^{l_1}_\epsilon$ denote the minimum time it takes state $\rho$ to evolve to another state at trace distance $\epsilon$.  Then, it is straightforward  to see that (see appendix)  $\tau^{l_1}_\epsilon(\rho)$ is lower bounded by
\begin{align} \label{QSL_l1}
\tau^{l_1}_\epsilon(\rho) \ge \frac{\epsilon}{F_H(\rho)} . 
\end{align}
where $F_H(\rho)$ is given by Eq.~(\ref{F_H}) and has been previously studied as a measure of asymmetry  \cite{Noether, Marvian_thesis}.
 This function also has a simple interpretation in terms of the speed of evolution. The fact that    
\begin{align}\label{speed}
F_H(\rho)=\left[{\frac{d}{dt} } {\|\rho-\rho(t)\|_1} \right]_{t=0^+} =\lim_{\epsilon\rightarrow 0^+} \frac{\epsilon}{\tau^{l_1}_\epsilon(\rho)}\ ,
\end{align}
implies that $F_H$ can be interpreted as the \emph{instantaneous} speed of evolution according to the trace distance. From this point of view,  the inequality of Eq.~(\ref{QSL_l1}) 
 is simply a bound on the average speed in terms of the instantaneous speed, and both of these notions of speed are measures of asymmetry relative to time translations. As we will discuss later, in the case of pure states, this bound reduces to  the Mandelstam-Tamm bound (up to a missing $\pi$ factor).

As our second example, we consider the  Renyi quantum  relative entropy,  Eq.~(\ref{def}), for $s=1/2.$
Let $\tau_{\epsilon}^{\text{Ren}}(\rho)$ be the minimum time $t$ it takes, under Hamiltonian $H$, for the relative Renyi entropy  $D_{1/2}(\rho, \rho(t))$ to become larger than or equal to $\epsilon$. Equivalently, $\tau_{\epsilon}^{\text{Ren}}(\rho)$   can be defined as the minimum time $t$ such that  $\text{tr}(\sqrt{\rho}\sqrt{\rho(t)})\le e^{-\frac{\epsilon}{2}}$. Then, as we show in the appendix, for any $\epsilon>0$ it holds that 
\begin{align}\label{QSL:Renyi}
\tau_{\epsilon}^{\text{Ren}}(\rho)\ge  \ \frac{ \ \sqrt{1-e^{-\epsilon/2}}}{ \sqrt{S_{H}(\rho)}} \ , \end{align}
where $S_H$
is the Wigner-Yanase skew information, defined in Eq. (\ref{S_H}).

The Wigner-Yanase skew information  has been shown to be a measure of asymmetry \cite{Noether} and has a simple interpretation in the context of quantum speed limits.  Noting that
\begin{equation}\label{accel}
 S_H(\rho)= \frac{1}{4} \left[\frac{d^2}{dt^2} D_\frac{1}{2}\left(\rho,\rho(t)\right) \right]_{{t=0}} = \frac{1}{2}  \lim_{\epsilon\rightarrow 0^+}\frac{\epsilon}{\left(\tau_{\epsilon}^{\text{Ren}}(\rho)\right)^2} 
\end{equation} 
we find that $S_H(\rho)$ can be interpreted as  (one fourth of) the \emph{instantaneous acceleration} of  evolution, relative to the Renyi relative entropy with $s=1/2$, at $t=0$. Since the \emph{instantaneous velocity}, i.e. the first derivative with respect to time, vanishes at $t=0$, from this point of view Eq.(\ref{QSL:Renyi}) is simply a bound on the average speed based on the instantaneous acceleration at $t=0$.

Both  functions $F_H(\rho)$ and $S_H(\rho)$ are zero if and only if  the state $\rho$ is incoherent, i.e., if and only if it commutes with $H$. They both capture the intuition that coherence of the state $\rho$ relative to the $H$ eigenbasis should be quantified by the \emph{noncommutativity} of $\rho$ and $H$, or in the case of $S_H$, the non-commutativity of $\sqrt\rho$ and $H$. Furthermore, they satisfy
\bes\label{measure}
\begin{align}
F_H(\rho)&\le 2 \Delta E(\rho) \ , \\
S_H(\rho)&\le \Delta E^2(\rho)\ ,
\end{align}
\ees
where both inequalities become equalities in the case of pure states. \footnote{One strategy to prove these bounds is the following.  First check them for the case of pure sates, which is straightforward.  Then, for a general mixed state, look at the purification of the state  and use the fact that by tracing over the purifying system,  measures of asymmetry do not increase. The latter monotonicity property follows from the fact that partial trace is a TI quantum operation.
} 

Using the fact that for pure states, the inequalities of Eq.~(\ref{measure}) hold as equalities, together with the fact $F_H$ and $S_H$ are non-increasing under TI quantum operations, we find that if, under a TI quantum operation, a pure state $\psi$ can be transformed to a pure state $\phi,$ then   $\Delta E(\psi)\ge \Delta E(\phi)$. In other words, the energy uncertainty $\Delta E$ is a measure of asymmetry in ``pure to pure" state transformations.

Two states are perfectly distinguishable if and only if  their trace distance is 2, and their relative Renyi entropy is $\infty$. This means that 
\begin{equation}
\tau_\perp(\rho)=\tau_{\infty}^{\text{Ren}}(\rho)=\tau^{l_1}_2(\rho)\ .
\end{equation}
Then, using the inequalities of Eq.~(\ref{measure}), we can see that   both bounds Eq.~(\ref{QSL_l1}) and Eq.~(\ref{QSL:Renyi}) reduce to  the original Mandelstam-Tamm bound, Eq.(\ref{QSL:MT}), up to a factor of $\pi$ in the case of Eq.(\ref{QSL_l1}), and a factor of $\pi/2$ in the case of  Eq.(\ref{QSL:Renyi}), which are irrelevant for any practical purposes.\footnote{The missing factors of $\pi$ and $\pi/2$ are due to the curvature of the space of pure states, which is not taken into account in the simple derivations of our bounds.}  

In the case of mixed states, however, the bounds of Eq.~(\ref{QSL_l1}) and Eq.~(\ref{QSL:Renyi}) can be much more powerful than the standard  Mandelstam-Tamm bound. In particular, unlike the Mandelstam-Tamm bound, these bounds correctly imply that for any incoherent state the speed of evolution is zero.   This is because they both satisfy the criterion we expressed in the previous section: the upper bounds on the speed evolution is a  measure of asymmetry relative to time translation, and so quantifies coherence relative to $H$. 

Another interesting property of the functions $S_H$ and $F_H$    is the fact that they are both convex, i.e., for $0\le p \le 1$ and for any pair of states $\rho$ and $\sigma$
\bes
\begin{align}
F_H\left(p \rho+ (1-p)\sigma \right)&\le  p\ F_H\left(\rho \right)+(1-p) F_H\left(\sigma \right)\\
S_H\left(p \rho+ (1-p)\sigma \right)&\le  p\ S_H\left(\rho \right)+(1-p)S_H\left(\sigma \right).
\end{align}
\ees
One way to see the convexity  of $F_H$ and $S_H$ is to use the fact they are,    respectively, instantaneous speed relative to trace distance, Eq. (\ref{speed}), and instantaneous acceleration  relative to relative Renyi entropy, Eq. (\ref{accel}), and then use the fact that trace distance and relative Renyi entropy are both jointly convex. \footnote{Convexity of skew information $S_H$  was shown originally by Wigner and Yanase  \cite{wigner1963information} and was one of their motivations to interpret the function $-S_H$ as an entropy.} 

As we discussed before, any function of state which quantifies the notion of speed of evolution is expected to be quasi-convex, and functions $F_H$ and $S_H$ have this property, while the  uncertainty function $\Delta E$ in the standard  Mandelstam-Tamm speed limit is not and therefore can  increase under mixing. 

One more appealing property which is satisfied by  Wigner-Yanase skew information $S_H$, but not by $F_H$,
  is additivity. 
Consider two (non-interacting) closed systems, $A$ and $B$, with Hamiltonians $H_A$ and $H_B$. The total Hamiltonian is  given by $H_\text{tot}=H_A\otimes I_B+I_A\otimes H_B$, where $I_A$ and $I_B$ are the identity operators on systems $A$ and $B$ respectively. Then, the Wigner-Yanase skew information is additive for uncorrelated (initial) joint states of $A$ and $B$, i.e.
\begin{align}
S_{H_\text{tot}}(\rho_A\otimes \rho_B)&=S_{H_A}(\rho_A)+S_{H_B}(\rho_B)\ .
\end{align}
Note that  the functions $\Delta E^2(\rho)$ and $E_\text{av}(\rho)-E_\text{min}(\rho)$, which show up in Mandelstam-Tamm and Margolus-Levitin bounds, are also additive.

\section{Margolus-Levitin bound and measures of asymmetry}\label{Levitin}

The next natural step is to understand the role of coherence and measures of asymmetry in the Margolus-Levitin QSL bound, Eq.(\ref{QSL:ML}). In the case of  the Mandelstam-Tamm  bound, Eq.(\ref{QSL:MT}), 
 we saw that the upper bound on the speed of evolution, the energy uncertainty  $\Delta E$, is itself 
a measure of time-translation asymmetry for pure states, which is to say that it is non-increasing in pure to pure state transformations that are achieved using  TI quantum operations. Hence for pure states, the standard Mandelstam-Tamm can be interpreted as an upper bound on a measure of asymmetry, namely, $\tau^{-1}_\perp$, by another measure of asymmetry, namely, $\Delta E$. Does the standard Margolus-Levitin bound have a similar interpretation?

In the following, we show that the answer is affirmative, and the function $E_\text{av}-E_\text{min}$ which shows up in the Margolus-Levitin QSL 
 is non-increasing  in pure to pure state transformations that are achieved using  TI quantum operations. 

Let $A_{\text{min}/\text{max}}(\rho)$ be the difference between $E_\text{av}(\rho)$, the average energy of state $\rho$, and $E_\text{min/max}(\rho)$, the minimum/maximum occupied energy level, i.e.
\bes
\begin{align}
A_{\text{min}}(\rho)\equiv E_\text{av}(\rho)-E_\text{min}(\rho)\ ,\\
A_{\text{max}}(\rho)\equiv E_\text{max}(\rho)-E_\text{av}(\rho)\ .
\end{align}
\ees
Then, one can easily see that:\\
\textbf{(i)} Functions $A_{\text{min/max}}$ are non-negative, i.e. $A_{\text{min/max}}(\rho)\ge 0$\ ,\\
\textbf{(ii)}  For a pair of systems that are non-interacting, which is to say that their total Hamiltonian is of the form $H_1\otimes I_2+I_1\otimes H_2$, 
the functions $A_{\text{min/max}}$ are additive, i.e.,
\beq
A_{\text{min/max}}(\rho_1\otimes\rho_2)=A_{\text{min/max}}(\rho_1)+A_{\text{min/max}}(\rho_2)\ .
\eeq 
\textbf{(iii)} For a pure state $\psi$, $A_{\text{min/max}}(\psi)$ is zero if (and only if) $\psi$ is invariant under time-translation, i.e., an eigenstate of the Hamiltonian.

Using these properties one can easily prove the following result:
\begin{theorem}\label{prop:Energy}
If there exists a  TI quantum operation under which a pure state $\psi$ evolves to a pure state $\phi$, then 
\beq
A_{\text{min/max}}(\phi)\le A_{\text{min/max}}(\psi)\ . 
\eeq
\end{theorem}
So, it follows that functions $A_{\text{min}}$ and $A_{\text{max}}$ are measures of asymmetry when restricted to pure states (according to Definition \ref{def1}). Note, however, that these functions can increase under TI quantum operations when evaluated on mixed states and thus they are not genuine measures of asymmetry.

 To prove this theorem, we use properties (i), (ii) and (iii) of the functions $A_{\text{min/max}}$.  We also make use of a version of Stinespring's dilation theorem, which  implies  that any symmetric quantum operation can be implemented using symmetric unitaries and symmetric pure states \cite{keyl1999optimal, Marvian_thesis}. More formally, it asserts that any TI quantum  operation $\mathcal{E}_\text{TI}$  (see Eq~(\ref{TI_op})),  can be implemented by coupling the system to an ancillary system, or \emph{environment}, with Hamiltonian $H_\text{env}$,  via a unitary $V_\text{TI}$ such that 
\begin{equation}
\mathcal{E}_\text{TI}(\rho)=\text{tr}_\text{env}\left(V_\text{TI} [\rho\otimes |E_0\rangle\langle E_0|] V_\text{TI}^\dag \right)\ ,
\end{equation}
where  the environment is initially in an eigenstate $|E_0\rangle$ of its Hamiltonian $H_\text{env}$, and the unitary $V_\text{TI}$ which couples the system and environment is an energy conserving unitary, i.e. 
\beq\label{cons}
\big[V_\text{TI}\ ,\ H\otimes I_\text{env}+I_\text{sys}\otimes H_\text{env}\big]=0\ .
\eeq

Suppose the TI quantum operation $\mathcal{E}_\text{TI}$ transforms the pure state $\psi$ to  the pure state $\phi$, and consider the Stinespring dilation of this operation. 
Given that the reduced state of the system must be $\phi$ at the end and given that evolution in the dilation is unitary, the joint state of the system and the environment at the end must be a pure product state, 
\begin{equation}
V_\text{TI} (|\psi\rangle\otimes |E_0\rangle)=|\phi\rangle\otimes |\theta_\text{env}\rangle
\end{equation}
where $|\theta_\text{env}\rangle$ is a pure state of the environment.

Next, we use the fact that $V_\text{TI}$ is an energy conserving unitary, i.e., it satisfies Eq.(\ref{cons}).  This implies that the energy distribution of the joint state of system and environment does not change after evolution, which in turn implies
\begin{equation}
A_\text{min/max} \big(|\psi\rangle\otimes |E_0\rangle\big)=A_\text{min/max}\big(|\phi\rangle\otimes |\theta_\text{env}\rangle\big)\ .
\end{equation}
Then, using the additivity of the functions $A_\text{min/max}$, we find
\begin{align}
A_\text{min/max} (\psi)&+A_\text{min/max} (|E_0\rangle) \ \  \nonumber \\ &=A_\text{min/max}(\phi)+A_\text{min/max}(\theta_\text{env})\ .
\end{align}
 But since $|E_0\rangle$ is an eigenstate of energy, we have $A_\text{min/max} (|E_0\rangle)=0$. Furthermore, using the fact $A_\text{min/max}$ is non-negative,  we have $A_\text{min/max}(\theta_\text{env})\ge 0$, and so
\begin{equation}
A_\text{min/max} (\psi)\ge A_\text{min/max} (\phi) \ .
\end{equation}
 This complete the proof. 
 
 Note that all the properties {(i)}, {(ii)} and {(iii)} of $A_\text{min/max} $ are also satisfied by the variance of the energy, $\Delta^2 E$, and consequently, using essentially the same argument, one can show that the variance is also non-increasing for pure to pure state transformations that are achieved by TI quantum opertations \cite{gour2008resource,MS11}.
 
Therefore, we have found that, just as for the Mandelstam-Tamm  QSL bound, the standard Margolus-Levitin QSL bound can  be interpreted as an upper bound on the speed of evolution, which is one measure of time-translation asymmetry, by a function that is also a measure of time-translation asymmetry in the case of pure states.

Finally, note that in the Margolus-Levitin QSL, we can replace $E_\text{av}(\psi)-E_\text{min}(\psi)$ by $E_\text{max}(\psi)-E_\text{av}(\psi)$, and the bound still holds, i.e. 
\begin{equation}
\tau^{-1}_\perp(\rho) \le \frac{2}{\pi} \left[E_\text{max}(\rho)-{E}_\text{av}(\rho)\right]  \ .
\end{equation}
This can be shown, for instance, by transforming  $H\rightarrow -H$ in the original bound.  This bound, however, is less useful in practice, because while physical Hamiltonians are bounded from below,  in general they do not have a bounded largest energy.

\section{Conclusion} 

In this paper, we discussed two different approaches for quantifying coherence which sometimes have been confused with each  other.  In the first approach, one considers  coherence as asymmetry relative to a group of translations such as time-translations or phase-shifts \cite{Modes, Noether, lostaglio2015quantum, lostaglio2015description}, whereas in the second approach, one considers coherence as the resource defined by incoherent operations  \cite{Coh_Plenio}. We have shown that only the first approach for quantifying  coherence is relevant in the context of quantum speed limits. This notion of coherence has also been shown to be relevant in the context of quantum thermodynamics \cite{lostaglio2015quantum, lostaglio2015description}, quantum metrology and quantum reference frames \cite{QRF_BRS_07, Marvian_thesis, Modes, MS11, GMS09}. 
We also showed that measures of coherence in the sense defined by  Baumgratz et. al. \cite{Coh_Plenio}, are a proper subset of measures of asymmetry. In particular, the Wigner-Yanase skew information is a measure of asymmetry which is not a measure of coherence based on the definition of Baumgratz et. al. \cite{Coh_Plenio}.

The notion of coherence as asymmetry relative to a group of translations naturally shows up in the context of quantum speed limits because the  speed of evolution is itself a measure of asymmetry relative to time-translations. 
This means that any function over states that can capture the notion of the speed of evolution of states should also be a measure of asymmetry.  Indeed one expects that a tight quantum speed limit  should bound the  speed of evolution with other measures of asymmetry. We have shown that the standard Mandelstam-Tamm   and Margolus-Levitin bounds satisfy this criterion in the case of pure states. Inspired by this intuition, we have found extensions of the Mandelstam-Tamm bound in which the speed of evolution is upper-bounded by genuine measures of asymmetry, such as the Wigner-Yanase skew information, which leads to significantly stronger bounds in the case of mixed states. A natural open question for future research is whether a similar goal can be achieved for the case of the Margolus-Levitin bound.


\textbf{Note:} During the last stages of preparing this manuscript,  we became aware of a  related work \cite{mondal2016quantum} recently posted on arXiv which uses a nice geometric argument to derive an extension of Mandelstam-Tamm bound based on Wigner-Yanase skew information.  This bound is basically  equivalent to one of  our bounds, i.e. Eq.(\ref{QSL:Renyi}), up to a factor of $\pi/2$.  Also, after posting the first version of this paper on arXiv, we became aware of another recent arXiv paper \cite{pires2015generalized}, which studies generalized geometric quantum speed limits based on the Petz contractive metrics, including  the Wigner-Yanase skew information. 

 \section{Acknowledgements}
We acknowledge helpful discussions with Gilad Gour and Barry Sanders. 
IM acknowledges support from grants ARO W911NF-12-1-0541 and NSF CCF-1254119. PZ acknowledges partial support from  ARO MURI Grant No. W911NF-11-1-0268.  RWS acknowledges the Perimeter Institute, which is supported by the Government of Canada through Industry Canada and by the Province of Ontario through the Ministry of Research and Innovation.
\\

 \bibliography{Ref_v12}

\begin{thebibliography}{55}
\expandafter\ifx\csname natexlab\endcsname\relax\def\natexlab#1{#1}\fi
\expandafter\ifx\csname bibnamefont\endcsname\relax
  \def\bibnamefont#1{#1}\fi
\expandafter\ifx\csname bibfnamefont\endcsname\relax
  \def\bibfnamefont#1{#1}\fi
\expandafter\ifx\csname citenamefont\endcsname\relax
  \def\citenamefont#1{#1}\fi
\expandafter\ifx\csname url\endcsname\relax
  \def\url#1{\texttt{#1}}\fi
\expandafter\ifx\csname urlprefix\endcsname\relax\def\urlprefix{URL }\fi
\providecommand{\bibinfo}[2]{#2}
\providecommand{\eprint}[2][]{\url{#2}}

\bibitem[{\citenamefont{Mandelstam and Tamm}(1945)}]{QSL_MT}
\bibinfo{author}{\bibfnamefont{L.}~\bibnamefont{Mandelstam}} \bibnamefont{and}
  \bibinfo{author}{\bibfnamefont{I.}~\bibnamefont{Tamm}}, \bibinfo{journal}{J.
  Phys.(USSR)} \textbf{\bibinfo{volume}{9}}, \bibinfo{pages}{1}
  (\bibinfo{year}{1945}).

\bibitem[{\citenamefont{{N. Margolus and L.B. Levitin}}(1998)}]{Margolus:98}
\bibinfo{author}{\bibnamefont{{N. Margolus and L.B. Levitin}}},
  \textbf{\bibinfo{volume}{120}}, \bibinfo{pages}{188} (\bibinfo{year}{1998}).

\bibitem[{\citenamefont{Levitin and Toffoli}(2009)}]{QSL_Toffoli}
\bibinfo{author}{\bibfnamefont{L.~B.} \bibnamefont{Levitin}} \bibnamefont{and}
  \bibinfo{author}{\bibfnamefont{T.}~\bibnamefont{Toffoli}},
  \bibinfo{journal}{Physical review letters} \textbf{\bibinfo{volume}{103}},
  \bibinfo{pages}{160502} (\bibinfo{year}{2009}).

\bibitem[{\citenamefont{Fleming}(1973)}]{QSL_Fleming}
\bibinfo{author}{\bibfnamefont{G.~N.} \bibnamefont{Fleming}},
  \textbf{\bibinfo{volume}{16}}, \bibinfo{pages}{232} (\bibinfo{year}{1973}),
  \urlprefix\url{http://dx.doi.org/10.1007/BF02819419}.

\bibitem[{\citenamefont{Bhattacharyya}(1983)}]{QSL_Bha83}
\bibinfo{author}{\bibfnamefont{K.}~\bibnamefont{Bhattacharyya}},
  \bibinfo{journal}{Journal of Physics A: Mathematical and General}
  \textbf{\bibinfo{volume}{16}}, \bibinfo{pages}{2993} (\bibinfo{year}{1983}).

\bibitem[{\citenamefont{Vaidman}(1992)}]{QSL_Vaidman}
\bibinfo{author}{\bibfnamefont{L.}~\bibnamefont{Vaidman}},
  \bibinfo{journal}{American journal of physics} \textbf{\bibinfo{volume}{60}},
  \bibinfo{pages}{182} (\bibinfo{year}{1992}).

\bibitem[{\citenamefont{Pfeifer}(1993)}]{Pfeifer:1993fk}
\bibinfo{author}{\bibfnamefont{P.}~\bibnamefont{Pfeifer}},
  \bibinfo{journal}{Physical Review Letters} \textbf{\bibinfo{volume}{70}},
  \bibinfo{pages}{3365} (\bibinfo{year}{1993}),
  \urlprefix\url{http://link.aps.org/doi/10.1103/PhysRevLett.70.3365}.

\bibitem[{\citenamefont{Pfeifer and Fr{\"o}hlich}(1995)}]{QSL_Pfeifer_RMP}
\bibinfo{author}{\bibfnamefont{P.}~\bibnamefont{Pfeifer}} \bibnamefont{and}
  \bibinfo{author}{\bibfnamefont{J.}~\bibnamefont{Fr{\"o}hlich}},
  \bibinfo{journal}{Reviews of Modern Physics} \textbf{\bibinfo{volume}{67}},
  \bibinfo{pages}{759} (\bibinfo{year}{1995}).

\bibitem[{\citenamefont{Giovannetti et~al.}(2003)\citenamefont{Giovannetti,
  Lloyd, and Maccone}}]{QSL_Lloyd}
\bibinfo{author}{\bibfnamefont{V.}~\bibnamefont{Giovannetti}},
  \bibinfo{author}{\bibfnamefont{S.}~\bibnamefont{Lloyd}}, \bibnamefont{and}
  \bibinfo{author}{\bibfnamefont{L.}~\bibnamefont{Maccone}},
  \bibinfo{journal}{Phys. Rev. A} \textbf{\bibinfo{volume}{67}},
  \bibinfo{pages}{052109} (\bibinfo{year}{2003}),
  \urlprefix\url{http://link.aps.org/doi/10.1103/PhysRevA.67.052109}.

\bibitem[{\citenamefont{Giovannetti et~al.}(2004)\citenamefont{Giovannetti,
  Lloyd, and Maccone}}]{QSL_Lloyd2}
\bibinfo{author}{\bibfnamefont{V.}~\bibnamefont{Giovannetti}},
  \bibinfo{author}{\bibfnamefont{S.}~\bibnamefont{Lloyd}}, \bibnamefont{and}
  \bibinfo{author}{\bibfnamefont{L.}~\bibnamefont{Maccone}},
  \bibinfo{journal}{Journal of Optics B: Quantum and Semiclassical Optics}
  \textbf{\bibinfo{volume}{6}}, \bibinfo{pages}{S807} (\bibinfo{year}{2004}),
  \urlprefix\url{http://stacks.iop.org/1464-4266/6/i=8/a=028}.

\bibitem[{\citenamefont{Zieli{\'n}ski and Zych}(2006)}]{QSL_Zych}
\bibinfo{author}{\bibfnamefont{B.}~\bibnamefont{Zieli{\'n}ski}}
  \bibnamefont{and} \bibinfo{author}{\bibfnamefont{M.}~\bibnamefont{Zych}},
  \bibinfo{journal}{Physical Review A} \textbf{\bibinfo{volume}{74}},
  \bibinfo{pages}{034301} (\bibinfo{year}{2006}).

\bibitem[{\citenamefont{Chau}(2010)}]{QSL_Chau}
\bibinfo{author}{\bibfnamefont{H.}~\bibnamefont{Chau}},
  \bibinfo{journal}{Physical Review A} \textbf{\bibinfo{volume}{81}},
  \bibinfo{pages}{062133} (\bibinfo{year}{2010}).

\bibitem[{\citenamefont{Deffner and
  Lutz}(2013{\natexlab{a}})}]{Deffner_uncertainty}
\bibinfo{author}{\bibfnamefont{S.}~\bibnamefont{Deffner}} \bibnamefont{and}
  \bibinfo{author}{\bibfnamefont{E.}~\bibnamefont{Lutz}},
  \bibinfo{journal}{Journal of Physics A: Mathematical and Theoretical}
  \textbf{\bibinfo{volume}{46}}, \bibinfo{pages}{335302}
  (\bibinfo{year}{2013}{\natexlab{a}}),
  \urlprefix\url{http://stacks.iop.org/1751-8121/46/i=33/a=335302}.

\bibitem[{\citenamefont{Taddei et~al.}(2013)\citenamefont{Taddei, Escher,
  Davidovich, and de~Matos~Filho}}]{QSL_Davidov_13}
\bibinfo{author}{\bibfnamefont{M.}~\bibnamefont{Taddei}},
  \bibinfo{author}{\bibfnamefont{B.}~\bibnamefont{Escher}},
  \bibinfo{author}{\bibfnamefont{L.}~\bibnamefont{Davidovich}},
  \bibnamefont{and}
  \bibinfo{author}{\bibfnamefont{R.}~\bibnamefont{de~Matos~Filho}},
  \bibinfo{journal}{Phys. Rev. Lett.} \textbf{\bibinfo{volume}{110}},
  \bibinfo{pages}{050402} (\bibinfo{year}{2013}),
  \urlprefix\url{http://link.aps.org/doi/10.1103/PhysRevLett.110.050402}.

\bibitem[{\citenamefont{del Campo et~al.}(2013)\citenamefont{del Campo,
  Egusquiza, Plenio, and Huelga}}]{QSL_Plenio_13}
\bibinfo{author}{\bibfnamefont{A.}~\bibnamefont{del Campo}},
  \bibinfo{author}{\bibfnamefont{I.}~\bibnamefont{Egusquiza}},
  \bibinfo{author}{\bibfnamefont{M.}~\bibnamefont{Plenio}}, \bibnamefont{and}
  \bibinfo{author}{\bibfnamefont{S.}~\bibnamefont{Huelga}},
  \bibinfo{journal}{Phys. Rev. Lett.} \textbf{\bibinfo{volume}{110}},
  \bibinfo{pages}{050403} (\bibinfo{year}{2013}),
  \urlprefix\url{http://link.aps.org/doi/10.1103/PhysRevLett.110.050403}.

\bibitem[{\citenamefont{Deffner and Lutz}(2013{\natexlab{b}})}]{QSL_Deffner_13}
\bibinfo{author}{\bibfnamefont{S.}~\bibnamefont{Deffner}} \bibnamefont{and}
  \bibinfo{author}{\bibfnamefont{E.}~\bibnamefont{Lutz}},
  \bibinfo{journal}{Phys. Rev. Lett.} \textbf{\bibinfo{volume}{111}},
  \bibinfo{pages}{010402} (\bibinfo{year}{2013}{\natexlab{b}}),
  \urlprefix\url{http://link.aps.org/doi/10.1103/PhysRevLett.111.010402}.

\bibitem[{\citenamefont{Zhang et~al.}(2014)\citenamefont{Zhang, Han, Xia, Cao,
  and Fan}}]{QSL_14_SciRep}
\bibinfo{author}{\bibfnamefont{Y.-J.} \bibnamefont{Zhang}},
  \bibinfo{author}{\bibfnamefont{W.}~\bibnamefont{Han}},
  \bibinfo{author}{\bibfnamefont{Y.-J.} \bibnamefont{Xia}},
  \bibinfo{author}{\bibfnamefont{J.-P.} \bibnamefont{Cao}}, \bibnamefont{and}
  \bibinfo{author}{\bibfnamefont{H.}~\bibnamefont{Fan}}, \bibinfo{journal}{Sci.
  Rep.} \textbf{\bibinfo{volume}{4}} (\bibinfo{year}{2014}),
  \urlprefix\url{http://dx.doi.org/10.1038/srep04890}.

\bibitem[{\citenamefont{Marvian and Lidar}(2015)}]{marvian2015quantum}
\bibinfo{author}{\bibfnamefont{I.}~\bibnamefont{Marvian}} \bibnamefont{and}
  \bibinfo{author}{\bibfnamefont{D.~A.} \bibnamefont{Lidar}},
  \bibinfo{journal}{Physical review letters} \textbf{\bibinfo{volume}{115}},
  \bibinfo{pages}{210402} (\bibinfo{year}{2015}).

\bibitem[{\citenamefont{Mondal and Pati}(2014)}]{mondal2014quantum}
\bibinfo{author}{\bibfnamefont{D.}~\bibnamefont{Mondal}} \bibnamefont{and}
  \bibinfo{author}{\bibfnamefont{A.~K.} \bibnamefont{Pati}},
  \bibinfo{journal}{arXiv preprint arXiv:1403.5182}  (\bibinfo{year}{2014}).

\bibitem[{\citenamefont{Marvian and Spekkens}(2014{\natexlab{a}})}]{Modes}
\bibinfo{author}{\bibfnamefont{I.}~\bibnamefont{Marvian}} \bibnamefont{and}
  \bibinfo{author}{\bibfnamefont{R.~W.} \bibnamefont{Spekkens}},
  \bibinfo{journal}{Phys. Rev. A} \textbf{\bibinfo{volume}{90}},
  \bibinfo{pages}{062110} (\bibinfo{year}{2014}{\natexlab{a}}),
  \urlprefix\url{http://link.aps.org/doi/10.1103/PhysRevA.90.062110}.

\bibitem[{\citenamefont{Marvian and Spekkens}(2014{\natexlab{b}})}]{Noether}
\bibinfo{author}{\bibfnamefont{I.}~\bibnamefont{Marvian}} \bibnamefont{and}
  \bibinfo{author}{\bibfnamefont{R.~W.} \bibnamefont{Spekkens}},
  \bibinfo{journal}{Nat Commun} \textbf{\bibinfo{volume}{5}}
  (\bibinfo{year}{2014}{\natexlab{b}}),
  \urlprefix\url{http://dx.doi.org/10.1038/ncomms4821}.

\bibitem[{\citenamefont{Marvian and Spekkens}(2016)}]{marvian2016quantify}
\bibinfo{author}{\bibfnamefont{I.}~\bibnamefont{Marvian}} \bibnamefont{and}
  \bibinfo{author}{\bibfnamefont{R.~W.} \bibnamefont{Spekkens}},
  \bibinfo{journal}{arXiv preprint arXiv:1602.08049}  (\bibinfo{year}{2016}).

\bibitem[{\citenamefont{Bartlett et~al.}(2007)\citenamefont{Bartlett, Rudolph,
  and Spekkens}}]{QRF_BRS_07}
\bibinfo{author}{\bibfnamefont{S.~D.} \bibnamefont{Bartlett}},
  \bibinfo{author}{\bibfnamefont{T.}~\bibnamefont{Rudolph}}, \bibnamefont{and}
  \bibinfo{author}{\bibfnamefont{R.~W.} \bibnamefont{Spekkens}},
  \bibinfo{journal}{Reviews of Modern Physics} \textbf{\bibinfo{volume}{79}},
  \bibinfo{pages}{555} (\bibinfo{year}{2007}).

\bibitem[{\citenamefont{Lostaglio
  et~al.}(2015{\natexlab{a}})\citenamefont{Lostaglio, Korzekwa, Jennings, and
  Rudolph}}]{lostaglio2015quantum}
\bibinfo{author}{\bibfnamefont{M.}~\bibnamefont{Lostaglio}},
  \bibinfo{author}{\bibfnamefont{K.}~\bibnamefont{Korzekwa}},
  \bibinfo{author}{\bibfnamefont{D.}~\bibnamefont{Jennings}}, \bibnamefont{and}
  \bibinfo{author}{\bibfnamefont{T.}~\bibnamefont{Rudolph}},
  \bibinfo{journal}{Physical Review X} \textbf{\bibinfo{volume}{5}},
  \bibinfo{pages}{021001} (\bibinfo{year}{2015}{\natexlab{a}}).

\bibitem[{\citenamefont{Baumgratz et~al.}(2014)\citenamefont{Baumgratz, Cramer,
  and Plenio}}]{Coh_Plenio}
\bibinfo{author}{\bibfnamefont{T.}~\bibnamefont{Baumgratz}},
  \bibinfo{author}{\bibfnamefont{M.}~\bibnamefont{Cramer}}, \bibnamefont{and}
  \bibinfo{author}{\bibfnamefont{M.~B.} \bibnamefont{Plenio}},
  \bibinfo{journal}{Phys. Rev. Lett.} \textbf{\bibinfo{volume}{113}},
  \bibinfo{pages}{140401} (\bibinfo{year}{2014}),
  \urlprefix\url{http://link.aps.org/doi/10.1103/PhysRevLett.113.140401}.

\bibitem[{\citenamefont{Marvian}(2012)}]{Marvian_thesis}
\bibinfo{author}{\bibfnamefont{I.}~\bibnamefont{Marvian}},
  \emph{\bibinfo{title}{Symmetry, Asymmetry and Quantum Information, PhD
  thesis,}} (\bibinfo{publisher}{University of Waterloo},
  \bibinfo{address}{https://uwspace.uwaterloo.ca/handle/10012/7088},
  \bibinfo{year}{2012}).

\bibitem[{\citenamefont{Lostaglio
  et~al.}(2015{\natexlab{b}})\citenamefont{Lostaglio, Jennings, and
  Rudolph}}]{lostaglio2015description}
\bibinfo{author}{\bibfnamefont{M.}~\bibnamefont{Lostaglio}},
  \bibinfo{author}{\bibfnamefont{D.}~\bibnamefont{Jennings}}, \bibnamefont{and}
  \bibinfo{author}{\bibfnamefont{T.}~\bibnamefont{Rudolph}},
  \bibinfo{journal}{Nature communications} \textbf{\bibinfo{volume}{6}}
  (\bibinfo{year}{2015}{\natexlab{b}}).

\bibitem[{\citenamefont{Marvian and Spekkens}(2013)}]{MS11}
\bibinfo{author}{\bibfnamefont{I.}~\bibnamefont{Marvian}} \bibnamefont{and}
  \bibinfo{author}{\bibfnamefont{R.~W.} \bibnamefont{Spekkens}},
  \bibinfo{journal}{New Journal of Physics} \textbf{\bibinfo{volume}{15}},
  \bibinfo{pages}{033001} (\bibinfo{year}{2013}).

\bibitem[{\citenamefont{Gour et~al.}(2009)\citenamefont{Gour, Marvian, and
  Spekkens}}]{GMS09}
\bibinfo{author}{\bibfnamefont{G.}~\bibnamefont{Gour}},
  \bibinfo{author}{\bibfnamefont{I.}~\bibnamefont{Marvian}}, \bibnamefont{and}
  \bibinfo{author}{\bibfnamefont{R.~W.} \bibnamefont{Spekkens}},
  \bibinfo{journal}{Physical Review A} \textbf{\bibinfo{volume}{80}},
  \bibinfo{pages}{012307} (\bibinfo{year}{2009}).

\bibitem[{\citenamefont{Gour and Spekkens}(2008)}]{gour2008resource}
\bibinfo{author}{\bibfnamefont{G.}~\bibnamefont{Gour}} \bibnamefont{and}
  \bibinfo{author}{\bibfnamefont{R.~W.} \bibnamefont{Spekkens}},
  \bibinfo{journal}{New Journal of Physics} \textbf{\bibinfo{volume}{10}},
  \bibinfo{pages}{033023} (\bibinfo{year}{2008}).

\bibitem[{\citenamefont{Marvian and Spekkens}(2014{\natexlab{c}})}]{MS_Short}
\bibinfo{author}{\bibfnamefont{I.}~\bibnamefont{Marvian}} \bibnamefont{and}
  \bibinfo{author}{\bibfnamefont{R.~W.} \bibnamefont{Spekkens}},
  \bibinfo{journal}{Physical Review A} \textbf{\bibinfo{volume}{90}},
  \bibinfo{pages}{014102} (\bibinfo{year}{2014}{\natexlab{c}}).

\bibitem[{\citenamefont{Brand{\~a}o and Gour}(2015)}]{brandao2015reversible}
\bibinfo{author}{\bibfnamefont{F.~G.} \bibnamefont{Brand{\~a}o}}
  \bibnamefont{and} \bibinfo{author}{\bibfnamefont{G.}~\bibnamefont{Gour}},
  \bibinfo{journal}{Physical review letters} \textbf{\bibinfo{volume}{115}},
  \bibinfo{pages}{070503} (\bibinfo{year}{2015}).

\bibitem[{\citenamefont{Coecke et~al.}(2014)\citenamefont{Coecke, Fritz, and
  Spekkens}}]{coecke2014mathematical}
\bibinfo{author}{\bibfnamefont{B.}~\bibnamefont{Coecke}},
  \bibinfo{author}{\bibfnamefont{T.}~\bibnamefont{Fritz}}, \bibnamefont{and}
  \bibinfo{author}{\bibfnamefont{R.~W.} \bibnamefont{Spekkens}},
  \bibinfo{journal}{arXiv preprint arXiv:1409.5531}  (\bibinfo{year}{2014}).

\bibitem[{\citenamefont{Fritz}(2015)}]{fritz2015mathematical}
\bibinfo{author}{\bibfnamefont{T.}~\bibnamefont{Fritz}},
  \bibinfo{journal}{arXiv preprint arXiv:1504.03661}  (\bibinfo{year}{2015}).

\bibitem[{\citenamefont{Yang and Chiribella}(2015)}]{yang2015optimal}
\bibinfo{author}{\bibfnamefont{Y.}~\bibnamefont{Yang}} \bibnamefont{and}
  \bibinfo{author}{\bibfnamefont{G.}~\bibnamefont{Chiribella}},
  \bibinfo{journal}{arXiv preprint arXiv:1502.00259}  (\bibinfo{year}{2015}).

\bibitem[{\citenamefont{Skotiniotis and Gour}(2012)}]{skotiniotis2012alignment}
\bibinfo{author}{\bibfnamefont{M.}~\bibnamefont{Skotiniotis}} \bibnamefont{and}
  \bibinfo{author}{\bibfnamefont{G.}~\bibnamefont{Gour}}, \bibinfo{journal}{New
  Journal of Physics} \textbf{\bibinfo{volume}{14}}, \bibinfo{pages}{073022}
  (\bibinfo{year}{2012}).

\bibitem[{\citenamefont{Narasimhachar and Gour}(2014)}]{narasimhachar2014phase}
\bibinfo{author}{\bibfnamefont{V.}~\bibnamefont{Narasimhachar}}
  \bibnamefont{and} \bibinfo{author}{\bibfnamefont{G.}~\bibnamefont{Gour}},
  \bibinfo{journal}{Physical Review A} \textbf{\bibinfo{volume}{89}},
  \bibinfo{pages}{033859} (\bibinfo{year}{2014}).

\bibitem[{\citenamefont{Toloui et~al.}(2011)\citenamefont{Toloui, Gour, and
  Sanders}}]{toloui2011constructing}
\bibinfo{author}{\bibfnamefont{B.}~\bibnamefont{Toloui}},
  \bibinfo{author}{\bibfnamefont{G.}~\bibnamefont{Gour}}, \bibnamefont{and}
  \bibinfo{author}{\bibfnamefont{B.~C.} \bibnamefont{Sanders}},
  \bibinfo{journal}{Physical Review A} \textbf{\bibinfo{volume}{84}},
  \bibinfo{pages}{022322} (\bibinfo{year}{2011}).

\bibitem[{\citenamefont{Vaccaro et~al.}(2008)\citenamefont{Vaccaro, Anselmi,
  Wiseman, and Jacobs}}]{vac2008}
\bibinfo{author}{\bibfnamefont{J.~A.} \bibnamefont{Vaccaro}},
  \bibinfo{author}{\bibfnamefont{F.}~\bibnamefont{Anselmi}},
  \bibinfo{author}{\bibfnamefont{H.~M.} \bibnamefont{Wiseman}},
  \bibnamefont{and} \bibinfo{author}{\bibfnamefont{K.}~\bibnamefont{Jacobs}},
  \bibinfo{journal}{Physical Review A} \textbf{\bibinfo{volume}{77}},
  \bibinfo{pages}{032114} (\bibinfo{year}{2008}).

\bibitem[{\citenamefont{Girolami and Yadin}(2015)}]{girolami2015witnessing}
\bibinfo{author}{\bibfnamefont{D.}~\bibnamefont{Girolami}} \bibnamefont{and}
  \bibinfo{author}{\bibfnamefont{B.}~\bibnamefont{Yadin}},
  \bibinfo{journal}{arXiv preprint arXiv:1509.04131}  (\bibinfo{year}{2015}).

\bibitem[{\citenamefont{Toloui and Gour}(2012)}]{toloui2012simulating}
\bibinfo{author}{\bibfnamefont{B.}~\bibnamefont{Toloui}} \bibnamefont{and}
  \bibinfo{author}{\bibfnamefont{G.}~\bibnamefont{Gour}}, \bibinfo{journal}{New
  Journal of Physics} \textbf{\bibinfo{volume}{14}}, \bibinfo{pages}{123026}
  (\bibinfo{year}{2012}).

\bibitem[{\citenamefont{Wigner and Yanase}(1963)}]{wigner1963information}
\bibinfo{author}{\bibfnamefont{E.~P.} \bibnamefont{Wigner}} \bibnamefont{and}
  \bibinfo{author}{\bibfnamefont{M.~M.} \bibnamefont{Yanase}},
  \bibinfo{journal}{Proceedings of the National Academy of Sciences of the
  United States of America} \textbf{\bibinfo{volume}{49}}, \bibinfo{pages}{910}
  (\bibinfo{year}{1963}).

\bibitem[{\citenamefont{Lieb}(1973)}]{lieb1973convex}
\bibinfo{author}{\bibfnamefont{E.~H.} \bibnamefont{Lieb}},
  \bibinfo{journal}{Advances in Mathematics} \textbf{\bibinfo{volume}{11}},
  \bibinfo{pages}{267} (\bibinfo{year}{1973}).

\bibitem[{\citenamefont{Wehrl}(1978)}]{wehrl1978general}
\bibinfo{author}{\bibfnamefont{A.}~\bibnamefont{Wehrl}},
  \bibinfo{journal}{Reviews of Modern Physics} \textbf{\bibinfo{volume}{50}},
  \bibinfo{pages}{221} (\bibinfo{year}{1978}).

\bibitem[{\citenamefont{Girolami}(2014)}]{Girolami}
\bibinfo{author}{\bibfnamefont{D.}~\bibnamefont{Girolami}},
  \bibinfo{journal}{Physical review letters} \textbf{\bibinfo{volume}{113}},
  \bibinfo{pages}{170401} (\bibinfo{year}{2014}).

\bibitem[{\citenamefont{Du and Bai}(2015)}]{du2015wigner}
\bibinfo{author}{\bibfnamefont{S.}~\bibnamefont{Du}} \bibnamefont{and}
  \bibinfo{author}{\bibfnamefont{Z.}~\bibnamefont{Bai}},
  \bibinfo{journal}{Annals of Physics} \textbf{\bibinfo{volume}{359}},
  \bibinfo{pages}{136} (\bibinfo{year}{2015}).

\bibitem[{\citenamefont{{C. W. Helstrom}}(1976)}]{Helstrom:book}
\bibinfo{author}{\bibnamefont{{C. W. Helstrom}}},
  \emph{\bibinfo{title}{{Quantum Detection and Estimation Theory}}}
  (\bibinfo{publisher}{{Academic Press}}, \bibinfo{address}{{New York}},
  \bibinfo{year}{1976}).

\bibitem[{\citenamefont{Hayashi}(2006)}]{hayashi2006quantum}
\bibinfo{author}{\bibfnamefont{M.}~\bibnamefont{Hayashi}},
  \emph{\bibinfo{title}{Quantum Information}} (\bibinfo{publisher}{Springer},
  \bibinfo{year}{2006}).

\bibitem[{\citenamefont{Petz}(1986)}]{petz1986quasi}
\bibinfo{author}{\bibfnamefont{D.}~\bibnamefont{Petz}},
  \bibinfo{journal}{Reports on mathematical physics}
  \textbf{\bibinfo{volume}{23}}, \bibinfo{pages}{57} (\bibinfo{year}{1986}).

\bibitem[{\citenamefont{M{\"u}ller-Lennert
  et~al.}(2013)\citenamefont{M{\"u}ller-Lennert, Dupuis, Szehr, Fehr, and
  Tomamichel}}]{muller2013quantum}
\bibinfo{author}{\bibfnamefont{M.}~\bibnamefont{M{\"u}ller-Lennert}},
  \bibinfo{author}{\bibfnamefont{F.}~\bibnamefont{Dupuis}},
  \bibinfo{author}{\bibfnamefont{O.}~\bibnamefont{Szehr}},
  \bibinfo{author}{\bibfnamefont{S.}~\bibnamefont{Fehr}}, \bibnamefont{and}
  \bibinfo{author}{\bibfnamefont{M.}~\bibnamefont{Tomamichel}},
  \bibinfo{journal}{Journal of Mathematical Physics}
  \textbf{\bibinfo{volume}{54}}, \bibinfo{pages}{122203}
  (\bibinfo{year}{2013}).

\bibitem[{\citenamefont{Wilde et~al.}(2014)\citenamefont{Wilde, Winter, and
  Yang}}]{wilde2014strong}
\bibinfo{author}{\bibfnamefont{M.~M.} \bibnamefont{Wilde}},
  \bibinfo{author}{\bibfnamefont{A.}~\bibnamefont{Winter}}, \bibnamefont{and}
  \bibinfo{author}{\bibfnamefont{D.}~\bibnamefont{Yang}},
  \bibinfo{journal}{Communications in Mathematical Physics}
  \textbf{\bibinfo{volume}{331}}, \bibinfo{pages}{593} (\bibinfo{year}{2014}).

\bibitem[{\citenamefont{Mosonyi}(2014)}]{mosonyi2014convexity}
\bibinfo{author}{\bibfnamefont{M.}~\bibnamefont{Mosonyi}},
  \bibinfo{journal}{arXiv preprint arXiv:1407.1067}  (\bibinfo{year}{2014}).

\bibitem[{\citenamefont{Keyl and Werner}(1999)}]{keyl1999optimal}
\bibinfo{author}{\bibfnamefont{M.}~\bibnamefont{Keyl}} \bibnamefont{and}
  \bibinfo{author}{\bibfnamefont{R.~F.} \bibnamefont{Werner}},
  \bibinfo{journal}{Journal of Mathematical Physics}
  \textbf{\bibinfo{volume}{40}}, \bibinfo{pages}{3283} (\bibinfo{year}{1999}).

\bibitem[{\citenamefont{Mondal et~al.}(2016)\citenamefont{Mondal, Datta, and
  Sazim}}]{mondal2016quantum}
\bibinfo{author}{\bibfnamefont{D.}~\bibnamefont{Mondal}},
  \bibinfo{author}{\bibfnamefont{C.}~\bibnamefont{Datta}}, \bibnamefont{and}
  \bibinfo{author}{\bibfnamefont{S.}~\bibnamefont{Sazim}},
  \bibinfo{journal}{Physics Letters A} \textbf{\bibinfo{volume}{380}},
  \bibinfo{pages}{689} (\bibinfo{year}{2016}).

\bibitem[{\citenamefont{Pires et~al.}(2015)\citenamefont{Pires, Cianciaruso,
  C{\'e}leri, Adesso, and Soares-Pinto}}]{pires2015generalized}
\bibinfo{author}{\bibfnamefont{D.~P.} \bibnamefont{Pires}},
  \bibinfo{author}{\bibfnamefont{M.}~\bibnamefont{Cianciaruso}},
  \bibinfo{author}{\bibfnamefont{L.~C.} \bibnamefont{C{\'e}leri}},
  \bibinfo{author}{\bibfnamefont{G.}~\bibnamefont{Adesso}}, \bibnamefont{and}
  \bibinfo{author}{\bibfnamefont{D.~O.} \bibnamefont{Soares-Pinto}},
  \bibinfo{journal}{arXiv preprint arXiv:1507.05848}  (\bibinfo{year}{2015}).

\end{thebibliography}

\appendix

\section{Proofs of bounds (\ref{QSL_l1}) and (\ref{QSL:Renyi}) and  proposition \ref{prop_incoh} }

\subsection{Proof of Inequality \ref{QSL_l1}}  

This follows from the fact that
\bes
\begin{align}
\left\| e^{-i H t}\rho e^{i H t}-\rho\right\|_1 &=\left\|\int_0^t ds\  \frac{d}{ds}\left( e^{-i H s}\rho e^{i H s}\right) \right\|_1\nonumber \\ &\le  \int_0^t ds \left\| \frac{d}{ds} (e^{-i H s}\rho e^{i H s}) \right\|_1
\nonumber \\ &= \int_0^t ds \left\| [H, e^{-i H s}\rho e^{i H s}] \right\|_1 \nonumber\\ &=   \int_0^t ds \left\| [H, \rho ] \right\|_1\nonumber\\ &= t \left\| [H, \rho ] \right\|_1\ ,
\end{align}
\ees
where to get the second line we have used the triangle inequality. So, if at time $t$ it holds that $\left\| e^{-i H t}\rho e^{i H t}-\rho\right\|_1=\epsilon $ then
\beq
\epsilon=\left\| e^{-i H t}\rho e^{i H t}-\rho\right\|_1  \le t  \left\| [H, \rho ] \right\|_1
\eeq
which proves bound \ref{QSL_l1}.
\subsection{Proof of Inequality \ref{QSL:Renyi}}  
The Renyi relative  entropy of $\rho$ and $\rho(t)=e^{-i H t} \rho e^{i Ht}$, for $s=1/2$ is given by 
\begin{align}
D_{1/2}(\rho,\rho(t))&=-2 \log \text{tr}\left(\sqrt{\rho}\sqrt{\rho(t)}\right)\nonumber \\ &=-2 \log \text{tr}\left(\sqrt{\rho}e^{-i H t}\sqrt{\rho}e^{i Ht}\right)
\end{align}
So $D_{1/2}(\rho,\rho(t))\ge\epsilon$ implies
\begin{align}
e^{-\frac{\epsilon}{2}}&\ge\text{tr}\left(\sqrt{\rho}\sqrt{e^{-i H t} \rho e^{i Ht}}\right) \nonumber\\
&=\text{tr}\left(\sqrt{\rho}e^{-i H t} \sqrt\rho e^{i Ht}\right)\nonumber\\ 
&=1+\int_0^t dr_1 \int_0^{r_1} dr_2\ \frac{\partial^2 }{\partial r^2_2} \text{tr}(\sqrt{\rho} e^{-i H r_2}\sqrt{\rho} e^{i H r_2}) \ ,
\end{align}
where to get the last line we have used the fact that  at $t=0$, first derivative of $\text{tr}\left(\sqrt{\rho}e^{-i H t} \sqrt\rho e^{i Ht}\right)$ with respect to $t$ vanishes, and so 
\beq
 \Big[\frac{\partial }{\partial r} \text{tr}(\sqrt{\rho} e^{-i H r}\sqrt{\rho} e^{i H r})\Big]_{r=t}=\int_0^t  dr\frac{\partial^2 }{\partial r^2} \text{tr}(\sqrt{\rho} e^{-i H r}\sqrt{\rho} e^{i H r}) .
\eeq 
So, we find
\begin{align}
&1-e^{-{\epsilon}/{2}}\le 1-\text{tr}\left(\sqrt{\rho}\ e^{-i H t} \sqrt\rho\ e^{i Ht}\right)\nonumber\\ &=-\int_0^t dr_1 \int_0^{r_1} dr_2\ \frac{\partial^2 }{\partial r^2_2}\ \text{tr}(\sqrt\rho e^{-i H r_2}\sqrt\rho e^{i H r_2})\nonumber \\ &\le\int_0^t dr_1 \int_0^{r_1} dr_2 \left|\frac{\partial^2 }{\partial r^2_2} \text{tr}(\sqrt\rho e^{-i H r_2}\sqrt\rho e^{i H r_2})\right|\nonumber  \\ &\le\int_0^t dr_1 \int_0^{r_1} dr_2 \ \  \underset{r_2\in [0,t]}{\max}\left|\frac{\partial^2 }{\partial r^2_2} \text{tr}(\sqrt\rho e^{-i H r_2}\sqrt\rho e^{i H r_2})\right|\nonumber  \\ &\le \frac{t^2}{2} \times \underset{r\in [0,t]}{\max}\left|\frac{\partial^2 }{\partial r^2} \text{tr}(\sqrt\rho e^{-i H r} \sqrt\rho e^{i H r})\right|\label{Eq-last2}\ ,
\end{align}
where to get the third line we have used the triangle inequality. Next, note that
\bes
\begin{align}
&\left|\frac{\partial^2 }{\partial r^2} \text{tr}(\sqrt\rho e^{-i H r} \sqrt\rho e^{i H r})\right|\nonumber\\ &=\left|\text{tr}(\sqrt\rho e^{-i H r} [H , [H,\sqrt\rho] ] e^{i H r})\right|\nonumber\\ &=\left|\text{tr}([H,\sqrt\rho] e^{-i H r}  [H,\sqrt\rho] e^{i H r})\right|\nonumber \\
&\le -\text{tr}\left([H,\sqrt\rho]  [H,\sqrt\rho] \right) \label{Sch}\\
&=2{S_{H}(\rho)}\ \label{Eq-last},
\end{align}
\ees
where to get Eq.(\ref{Sch}) we have used Cauchy-Schwartz inequality. Combining Eq.(\ref{Eq-last}) and Eq.(\ref{Eq-last2}) we find
\begin{align}
1-e^{-{\epsilon}/{2}}\le  t^2\ S_{H}(\rho)\ ,
\end{align}
which completes the proof.

\subsection{Proof of proposition \ref{prop_incoh}}

Recall that according to Baumgratz et. al. \cite{Coh_Plenio}, incoherent operations are quantum operations for which a Kraus decomposition $\mathcal{E}(\cdot)=\sum_\mu K_\mu(\cdot) K^\dag_\mu $ exists such that for each Kraus operator $K_\mu$ and any \emph{incoherent state} $\rho$,  $K_\mu\rho K^\dag_\mu/\text{tr}(K_\mu\rho K^\dag_\mu)$  is also an incoherent state \cite{Coh_Plenio}. We assume incoherent states are  states which are diagonal in the eigenbasis of the observable $H$. 

As we show in the following,  any TI quantum operation $\mathcal{E}_\TI$, i.e. any quantum operation satisfying $e^{-i H t}\mathcal{E}_{\text{TI}}(\cdot) e^{i H t}=\mathcal{E}_\text{TI}\left(e^{-i H t}\cdot e^{i H t}\right)$ for all $t\in\mathbb{R}$, has a Kraus decomposition as $\mathcal{E}_\TI(\cdot)=\sum_\mu K_\mu(\cdot)K^\dag_\mu$,  where each Kraus operator $K_\mu$ satisfies
\beq\label{Kraus}
e^{-i H t} K_\mu e^{i H t}= e^{i\omega_\mu t} K_\mu\ ,\ \ \ \forall t\in\mathbb{R}
\eeq
for a real number $\omega_\mu$. Assuming this equation, it is straightforward to show that $\mathcal{E}_\TI$ is an incoherent operation in the sense of  Baumgratz et. al. \cite{Coh_Plenio}: For any incoherent state $\rho_\TI\in \mathcal{I}_H$ and  for any $t\in\mathbb{R}$ it holds that
\bes
\begin{align}
K_\mu \rho_\TI K^\dag_\mu&= K_\mu \big(e^{-i Ht }\rho_\TI e^{i Ht }\big) K^\dag_\mu\\ &= e^{-i Ht }\big( K_\mu \rho_\TI  K^\dag_\mu \big)e^{i Ht }\ ,
\end{align}
\ees
where the fist equality follows from the fact that $\rho_\TI$ is incoherent, and the second equality follows from Eq.(\ref{Kraus}).
 Since this holds for all $t\in\mathbb{R}$, it follows that $K_\mu \rho_\TI  K^\dag_\mu$ commutes with $H$, and hence is incoherent in $H$ eigenbasis. Thus  $K_\mu \rho_\TI  K^\dag_\mu/\text{tr}(K_\mu \rho_\TI  K^\dag_\mu)$ is an incoherent state. Since this  holds for arbitrary incoherent state $\rho_\TI$, it follows that $\mathcal{E}_\TI$ is an incoherent operation according to the definition of Baumgratz et. al. \cite{Coh_Plenio}.  
 
 Thus to complete the proof we only need to show Eq.(\ref{Kraus}). This equation is indeed a special case of lemma 1 of  \cite{gour2008resource}. For completeness, here we present a different proof of this fact based on the Steinspring representation of symmetric operations  \cite{keyl1999optimal}, which we also used in section \ref{Levitin}.
 
Any TI quantum  operation $\mathcal{E}_\text{TI}$  can be implemented by coupling the system to an ancillary system, or \emph{environment}   via a unitary $V_\text{TI}$ such that 
\begin{equation}
\mathcal{E}_\text{TI}(\rho)=\text{tr}_\text{env}\left(V_\text{TI} [\rho\otimes |E_0\rangle\langle E_0|] V_\text{TI}^\dag \right)\ ,
\end{equation}
where  the environment is initially in an eigenstate $|E_0\rangle$ of   $H_\text{env}$ with eigenvalue $E_0$, and the unitary $V_\text{TI}$ which couples the system to the environment satisfies 
\beq\label{cons2}
[V_\text{TI}, H\otimes I_\text{env}+I_\text{sys}\otimes H_\text{env}]=0\ .
\eeq
Let $\{|E_l\rangle\}$ be the orthonormal set of eigenvectors of $H_\text{env}$, such that $H_\text{env} |E_l\rangle=E_l|E_l\rangle$ (To simplify the notation we assume there is no degeneracy). Then a Kraus decomposition of $\mathcal{E}_\text{TI}$ is given by $\mathcal{E}_\text{TI}(\cdot)=\sum_l K_l(\cdot)K_l^\dag$, where
\beq
K_l= \langle E_l|V_\text{TI}|E_0\rangle\ .
\eeq
It can be easily seen that for any Kraus operator $K_\mu$ it holds that
\begin{align}
e^{-i H t} K_l e^{i H t}&= e^{-i H t} \langle E_l |V_\text{TI}|E_0\rangle e^{i H t}\nonumber\\ &=e^{-i H t} e^{i E_l t}  \langle E_l |e^{-i H_\text{env} t} V_\text{TI}|E_0\rangle e^{i H t} 
\nonumber\\ &= e^{i E_l t}  \langle E_l | \big(e^{-i H t}\otimes e^{-i H_\text{env} t} \big) V_\text{TI}|E_0\rangle e^{i H t} 
\nonumber\\ &= e^{i E_l t}  \langle E_l | V_\text{TI}  \big(e^{-i H t}\otimes e^{-i H_\text{env} t} \big) |E_0\rangle  e^{i H t} 
\nonumber\\ &= e^{i (E_l-E_0) t}  \langle E_l | V_\text{TI}   |E_0\rangle \nonumber\\ 
&= e^{i (E_l-E_0) t} K_l\  .
\end{align}
It follows that any TI operation $\mathcal{E}_\text{TI}$ has a Kraus decomposition satisfying Eq.(\ref{Kraus}).  
(This argument also provides an interpretation of constants $\omega_\mu$ in Eq.(\ref{Kraus}): In the case where the generator $H$ is the system Hamiltonian, and  $\mathcal{E}_\text{TI}$ is invariant under time translation, constant  $\omega_\mu$ is the energy transferred  from the  environment to the system, given that the process corresponding to Kraus operator $K_\mu$ has happened.)

\end{document}